\documentclass{article}
\usepackage{amsmath,amssymb,amsfonts,amsthm,bm}
\usepackage{mathrsfs}
\usepackage{indentfirst}
\usepackage[pdftex]{color,graphicx}
\usepackage{subfig}
\usepackage[pdftex,bookmarks,unicode,colorlinks]{hyperref}
\usepackage{enumerate} 
\usepackage[numbers]{natbib} 
\usepackage{yfonts} 
\usepackage{caption} 
\usepackage{mathtools}
\usepackage{chngcntr}
\usepackage{algorithmicx,algorithm} 
\usepackage{lineno}
\usepackage{bbm}

\theoremstyle{plain}
\newtheorem{theorem}{Theorem}[section]
\newtheorem{lemma}[theorem]{Lemma}
\newtheorem{corollary}[theorem]{Corollary}
\newtheorem{proposition}[theorem]{Proposition}

\theoremstyle{remark}
\newtheorem{remark}[theorem]{Remark}
\newtheorem{example}[theorem]{Example}

\theoremstyle{definition}

\newtheorem{assumption}{Assumption}

\title{Quantum Algorithms for Stochastic Differential Equations: \\A Schr\"odingerisation Approach}
\author{\bf\normalsize{ Shi Jin$^{1,2,3,4}$, Nana Liu$^{1,2,3,4,5}$, Wei Wei$^{1}$ \footnote{\texttt{weiw\_sjtu@sjtu.edu.cn}}}\\
\footnotesize{$^1$\textit{Institute of Natural Sciences, Shanghai Jiao Tong University, Shanghai 200240, China}} \\
\footnotesize{$^2$\textit{School of Mathematical Sciences, Shanghai Jiao Tong University, Shanghai 200240, China}} \\ 
\footnotesize{$^3$\textit{Ministry of Education Key Laboratory in Scientific and Engineering Computing,}} \\
\footnotesize{\textit{Shanghai Jiao Tong University, Shanghai 200240, China}} \\
\footnotesize{$^4$\textit{Shanghai Artificial Intelligence Laboratory, Shanghai, China}} \\
\footnotesize{$^5$\textit{University of Michigan-Shanghai Jiao Tong University Joint Institute, Shanghai 200240, China.}}}

\date{}

\begin{document}

\maketitle
\vspace{-0.3in}

\begin{abstract}
    Quantum computers are known for their potential to achieve up-to-exponential speedup compared to classical computers for certain problems. 
    To exploit the advantages of quantum computers, we propose quantum algorithms for linear stochastic differential equations, utilizing the Schr\"odingerisation method for the corresponding approximate equation by treating the noise term as a (discrete-in-time) forcing term.
    Our algorithms are applicable to stochastic differential equations with both Gaussian noise and $\alpha$-stable L\'evy noise. The gate complexity of our algorithms exhibits an $\mathcal{O}(d\log(Nd))$ dependence on the dimensions $d$ and sample sizes $N$, where its corresponding classical counterpart requires nearly exponentially larger complexity in scenarios involving large sample sizes. 
    In the Gaussian noise case, we show the strong convergence of first order in the mean square norm for the approximate equations.
    The algorithms are numerically verified for the Ornstein–Uhlenbeck processes, geometric Brownian motions, and one-dimensional L\'evy flights.
    Schr\"odingerisation $\cdot$ Strong convergence
\end{abstract}

\section{Introduction}

Stochastic differential equations (SDEs) finds applications in various fields including physics, climate modeling, finance and machine learning \cite{SDENet,song2021scorebased,NeurIPS2020SGD} . Their broad applicability stems from their ability to encapsulate the randomness in nature, providing insights into the dynamics of complex systems with uncertainty. For example, the stochastic Langevin equation and its variants are widely applied in physics and machine learning, including generative AI \cite{2021NeurIPSSchrodingerBridge,2022ICLRLangevin,2021NeurIPSvariational,song2021scorebased}. Another important application of SDEs is in Monte Carlo simulations, where SDEs are used to model and simulate random processes over time. 

For many applications, it becomes interesting to develop quantum algorithms to leverage the potential advantage offered by quantum computers. Quantum computers, which are different from their classical counterparts, use qubits instead of bits. 
Some of the quantum algorithms have already proved to have quantum advantages over their classical counterparts in, for example, searching \cite{Grover}, cryptography \cite{Shor}, sampling \cite{BosonSampling}, solving linear algebraic equations \cite{ChildsQASL2017,HHL} and etc. Quantum algorithms also have applications in computational finance, which is also closely related to SDEs \cite{QCforfinancePRR,VariationalQCforSDE,QCforfinanceReview}.

There are quantum algorithms for stochastic processes that have  demonstrated quantum advantage. For example,
as an analogue of classical random walk, the quantum walk achieves up-to exponential speedup \cite{ExponentialQW} and can serve as a universal computational primitive \cite{UniversalQW} for simulating arbitrary quantum circuits. It finds wide applications in Hamiltonian simulation, quantum information processing, and search algorithms \cite{QWreview}.  
Quantum algorithms have also been proposed for Markov Chain Monte Carlo simulation,  commonly used to sample from the Boltzmann distribution \cite{QuanEnhanceMCMC}. They potentially ease the computational bottlenecks posed by this sampling problem and achieve quantum speedup in optimization \cite{QAMCMCopt2025}. Quantum algorithms can also provide nearly quadratic speedup to Monte Carlo methods for computing expected values of random variables \cite{QAMLMC,QAMC,QAMCinfinitevariance}. 
The Lindbladian equation, which describes an open quantum system, can be interpreted as the master equation for the reduced density matrix of a quantum stochastic differential equation. Classical numerical schemes for SDEs also enable quantum computing methods for the Lindbladian equation to achieve arbitrarily high-order accuracy \cite{Xiantao2024OpenQS}.

Although classical schemes for SDEs, such as the Euler-Maruyama scheme, have been proposed for decades and are widely used in various areas, the quantum algorithms for SDEs are underdeveloped.
The motivation for designing quantum algorithms for SDEs arises from the fact that, while simulating a single sample path incurs only a manageable computational cost, achieving a target precision $\epsilon$ with the Monte Carlo method requires at least $\mathcal{O}(1/\epsilon^2)$ simulations, leading to substantial computational expense \cite{QAMLMC}.
Several quantum algorithms have been proposed for SDEs and Monte Carlo simulations. Reference \cite{VariationalQCforSDE} develops a quantum-classical hybrid algorithm based on variational quantum simulation. It discretizes the SDE into a trinomial tree model and computes the nodes of the tree by time-evolution of a quantum state. The nodes of the tree of each level represent the distribution of the SDE at each time step. However, the tree model itself suffers from insufficient accuracy. Reference \cite{QAMC} guarantees the near-quadratic speedup of quantum computers in the sampling process of Monte Carlo methods. Based on this, reference \cite{QAMLMC} develops quantum algorithms for multi-level Monte Carlo simulation with simulation of SDEs. However, it is developed based on oracles for the expectation of SDEs, and does not offer a concrete quantum algorithm. 
Reference \cite{Xiantao2024OpenQS} uses the relationship between open quantum systems, SDEs, and Hamiltonian simulations to develop high order quantum algorithms for the open quantum system. 
Thus, quantum algorithms for SDEs that both match the accuracy of classical schemes and offers greater efficiency, are of significant interest.
 
Recently, a Schr\"odingerisation approach \cite{Schrodingerisation1,jin2023quantum} has been proposed to solve general linear ordinary and partial differential equations via Hamiltonian simulation, which can also be implemented  in quantum circuits \cite{hu2024qc,jin2024heatqc} and applicable to a wide variety of differential equations \cite{Hu2024Multiscale,Ma2023Maxwell,Ma2024inhomov2,Ma2024illpoesd,jin2024qsFKP}. The method of Schr\"odingerisation involves the incorporation of an auxiliary variable to transform a linear differential equation into an equation analogous to the Schr\"odinger equation, with unitary evolution process, in one higher dimension. Subsequently, the transformed equation lends itself to implementation on quantum computers through Hamiltonian simulation. The desired result is recovered from the Schr\"odingerised equation on a properly chosen interval of the auxiliary variable. This method maintains the major structure of the original classical differential equations,  thus, consequently, one can develop quantum algorithms that retain the convergence rate of classical schemes while leveraging the benefits of quantum computing.

In this paper, we propose a general quantum algorithm to simulate the solutions of linear SDEs, including SDEs with additive noise and multiplicative noise. The main idea is to treat the noise term as a discrete-in-time forcing term and then uses the  Schr\"odingerisation method for inhomogeneous differential equations \cite{Ma2024inhomov2}.  While most quantum algorithms focus on SDEs with Gaussian noise, our approach is applicable to SDEs driven by stable L\'evy processes, which is a more general type of noise. We design the quantum algorithms for the approximate equation through the Schr\"odingerisation method, which approximates the corresponding SDEs under mean square norm in the Gaussian noise case. We also discuss the important issue of the choice of the recovery interval and its influence on the precision of the final result. Theoretically, for both the additive and multiplicative Gaussian noise cases, we prove the strong convergence of first order for the approximate solution.
The quantum advantage of our algorithm is exhibited by its gate complexity of  $\mathcal{O}(d\log(Nd))$ dependence on the dimensions $d$ and the sample size $N$, which is a  nearly exponential speedup compared to its classical counterparts.
We illustrate the validation of the algorithms to the Ornstein–Uhlenbeck processes, geometric Brownian motions and one-dimensional L\'evy flights by numerical experiments.

This paper is organized as follows: Section \ref{Schro} introduces the approximate equations for both SDEs with additive noise and multiplicative noise. By the Sch\"odingerisation method, we provide the technical details of the quantum algorithms for SDEs. We give the gate complexity of our quantum algorithm in subsection \ref{complexity} and provide a complexity analysis which demonstrates  the quantum advantage for dimensions and  sample sizes. In Section \ref{examp}, we show the correctness of our method by numerical experiments for a one dimensional Ornstein–Uhlenbeck process, a one dimensional geometric Brownian motion and a one dimensional L\'evy flight. In Section \ref{convrate}, we show the convergence rate of the approximate equations in the mean square sense. Together with the error of the recovery part, we give the overall error of our algorithms. We summarize our method and discuss further applications in Section \ref{summary}.

\section{Quantum algorithms for SDEs}\label{Schro}

In this section, we will introduce the quantum algorithms for linear SDEs with Gaussian noise. The algorithms are based on the combination of the Euler-Maruyama scheme and the Schr\"odingerisation technique \cite{Schrodingerisation1,jin2023quantum}.

\subsection{ Linear SDEs with additive Gaussian noise}\label{SDEadditive}

Given a probability space $(\Omega, \mathcal{F},\mathbb{P})$, we consider the following $d$-dimensional SDE:
\begin{equation}\label{AddSDE}
    d X(t) =A X(t) \,dt+B\, dW_t, \quad X(0)=x_0, \quad 0\leq t \leq T.
\end{equation}
Here, $W_t$ is the $d$-dimensional standard Brownian motion. $A$ and $B$ are $d \times d$ matrices.

The Euler-Maruyama scheme for equation \eqref{AddSDE} is described as follows.

Let $t_k=kT/N_T, \Delta t=T/N_T$, and $\{\widehat{X}(t_k)\}_{0\leq k \leq N}$ be the numerical solution to \eqref{AddSDE}:
\begin{equation}\label{EMscheme}
    \widehat{X}(t_{k+1})=\widehat{X}(t_k)+A \widehat{X}(t_k) \Delta t+B \Delta W_k, \quad \widehat{X}(t_0)=x_0 \textrm{ and } k=0,1,\dots, N_T-1
\end{equation}
where $\Delta W_k=W_{t_{k+1}}-W_{t_k}$ is a Gaussian random variable with mean $0$ and covariance matrix $\Delta t I$.

We first replace the discrete time in the Euler-Maruyama method with a continuous time variable, so it becomes an ordinary differential equation for each time interval. More precisely, for the temporal mesh $\{t_k=k \Delta t: \Delta t =T/N_T,k=0,1,\cdots, N_T\}$, we aim to calculate the solution of the following ordinary differential equation in the $k$-th iteration:
\begin{align}\label{NumForm}
    \left\{
        \begin{aligned}
            & \frac{d }{dt} \widetilde{X}(t)= A \widetilde{X}(t)+\frac{B \Delta W_k}{\Delta t}, \quad t_k < t \leq t_{k+1}, \quad k \geq 1,\\
            & \textrm{with }\widetilde{X}(t_k) \textrm{ calculated in the } (k-1) \textrm{-th iteration.}
        \end{aligned}
    \right.
\end{align}
Here $\widetilde{X}(t_0)=x_0$ and $\Delta W_k=W_{t_{k+1}}-W_{t_k}$ which can be simulated by a group of independent $d$-dimensional standard Gaussian random variables $\{\boldsymbol{\xi}_k: \boldsymbol{\xi}_k\sim N(0,\boldsymbol{I})\}$ :
\begin{equation*}
    \Delta W_k \overset{d}{=} \sqrt{\Delta t}\boldsymbol{\xi_k}.
\end{equation*}
Now we can implement the Schr\"odingerisation approach \cite{jin2023quantum} to \eqref{NumForm}. Take 
\begin{equation*}
    \widetilde{A}_k=\left(\begin{array}{cc}
    A &B \Delta W_k/\sqrt{\Delta t}\\
    \boldsymbol{0}^\top &0
    \end{array} \right), \quad \widetilde{Y}=\left( \begin{array}{c}
         \widetilde{X}(t)  \\
         1/\sqrt{\Delta t}
    \end{array}\right),
\end{equation*}
then equation \eqref{NumForm} becomes,
\begin{align}\label{TD}
    \left\{
        \begin{aligned}
            \frac{d }{dt} \widetilde{Y}(t) & = \widetilde{A}_k \widetilde{Y}(t), \quad t_k < t \leq t_{k+1}, \\
            \widetilde{Y}(0) & =(x_0,1/\sqrt{\Delta t})^\top . 
        \end{aligned}
    \right.
\end{align}
For every $\widetilde{A}_k$, we decompose it into a Hermitian matrix $H_{1,k}$ and an anti-Hermite matrix $iH_{2,k}$.
\begin{equation*}
    \widetilde{A}_k=H_{1,k}+iH_{2,k},
\end{equation*}
where 
\begin{equation*}
    H_{1,k}=\frac{\widetilde{A}_k+\widetilde{A}_k^\dag}{2}=H_{1,k}^{\dag}, \quad H_{2,k}=\frac{\widetilde{A}_k-\widetilde{A}_k^\dag}{2i}=H_{2,k}^{\dag}.
\end{equation*}
Apply the warped phase transformation to $\widetilde{Y}$, that is taking $\boldsymbol{v}(t,p)=e^{-p}\widetilde{Y}(t)$ for $p>0$. Then, $\boldsymbol{v}(t,p)$ satisfies the $(d+1)$-dimensional equation
\begin{align}\label{Addwptrans}
    \left\{
        \begin{aligned}
            \frac{d }{dt} \boldsymbol{v}(t,p) & = -H_{1,k} \partial_p \boldsymbol{v}(t,p) +iH_{2,k}\boldsymbol{v}(t,p), \quad t_k < t\leq t_{k+1},\\
            \boldsymbol{v}(0,p) & =e^{-|p|}\widetilde{Y}(0), \quad p\in(-\infty,\infty). 
        \end{aligned}
    \right.
\end{align}
By taking the Fourier transform of $\boldsymbol{v}(t,p)$ with respect to $p$ and denoting it as $w(t,\eta)$, we obtain 
\begin{align}\label{TH}
    \left\{
        \begin{aligned}
            \frac{d }{dt} w(t,\eta) & = -i \eta H_{1,k} w(t,\eta) + i H_{2,k} w(t,\eta), \quad t_k < t\leq t_{k+1},\\
            w(0,\eta) & =\frac{1}{\pi(1+\eta^2)} \widetilde{Y}(0).
        \end{aligned}
    \right.
\end{align}

Here we have derived an equation \eqref{TH} analogous to the Schr\"odinger equation. 

Let $H_1(t)=\sum_{k=0}^{N_T-1}H_{1,k}\mathbbm{1}_{[t_k,t_{k+1})}(t)$. 
If all eigenvalues of $H_1(t)$ are negative, the solution $\widetilde{Y}(t)$ can be restored by the integration:
\begin{equation*}
    \widetilde{Y}(t)= \int_0^\infty \boldsymbol{v}(t,p)dp.
\end{equation*}
 If $H_{1}(t)$ contains some positive eigenvalues, which will be the case with the noise term, then $\widetilde{Y}(t)$ can still be restored through the normalized integration formula:
\begin{equation}\label{NormIntExpli}
    \widetilde{Y}(t)=\frac{\int_{p^*}^{p^{*}+R}\boldsymbol{v}(t,p) dp}{\int_{p^*}^{p^{*}+R} e^{-p} dp},
\end{equation}
where $p^{*}+R>p^* \geq\lambda^+_{\text{max}}(H_{1})T$, with $\lambda_{\max}^{+}(H_1) =\max \{ \sup_{0<t<T} \{|\lambda|: \lambda \in \sigma(H_1(t)), \lambda >0 \}, 0\}$\cite[Theorem 3.1]{Ma2024inhomov2}.
\medskip

In Section \ref{convrate}, we will show in Proposition \ref{Addconvrate} that the mean square norm between $\widetilde{X}$ and $X$ satisfies
\begin{equation*}
    \|\widetilde{X}(T)-X(T)\|_{L^2} \lesssim \Delta t.
\end{equation*}

\begin{remark}
    Equation \eqref{TH} can be seen as a piecewise constant time-dependent Schr\"odinger equation with Hamiltonian $\widetilde{H}(t)$ defined by
    \begin{equation*}
        \widetilde{H}(t)=\sum_{k=0}^{N_T-1}\big (\eta H_{1,k}- H_{2,k} \big )\mathbbm{1}_{(t_k,t_{k+1}]}(t),
    \end{equation*}
where $H_{1,k}$ and $H_{2,k}$ are linear Hermitian random matrix. The random matrix $H_{1,k}$ and $H_{2,k}$ have the following explicit expression:
    \begin{equation*}
        H_{1,k}=\left(\begin{array}{cc}
            \frac{1}{2}(A+A^\dag) &\frac{1}{2}B \Delta W_k/\sqrt{\Delta t}  \\
            \frac{1}{2}(\Delta W_k^\top/\sqrt{\Delta t}) B^\dag & 0
        \end{array} \right), \quad 
        H_{2,k}=\left(\begin{array}{cc}
            \frac{1}{2i}(A-A^\dag) & -\frac{i}{2} B \Delta W_k/\sqrt{\Delta t} \\
            \frac{i}{2}  (\Delta W_k^\top/\sqrt{\Delta t}) B^\dag & 0
        \end{array} \right).
    \end{equation*}
\end{remark}

\begin{remark}
    The Schr\"odingerisation approach can also be applied to linear SDEs driven by $\alpha$-stable L\'evy processes with $1<\alpha<2$. Specifically, consider a $d$-dimensional linear SDE driven by isotropic $\alpha$-stable L\'evy processes:
    \begin{equation}\label{SchemeLevy}
        d X(t) =A X(t) \,dt+B dL^\alpha_t, \quad X(0)=x_0, \quad 0\leq t \leq T,
    \end{equation}
    where $L_t^\alpha$ is a $d$-dimensional isotropic $\alpha$-stable L\'evy process with $1<\alpha<2$ and it satisfies $L^\alpha_t \overset{d}{=}t^{\frac{1}{\alpha}}L_1$. Then we can proceed with the approximate equation
    \begin{align}\label{NumFormLevy}
        \left\{
            \begin{aligned}
                & \frac{d }{dt} \widetilde{X}(t)= A \widetilde{X}(t)+\frac{B \Delta L^\alpha_k}{\Delta t}, \quad t_k < t \leq t_{k+1}, \quad k \geq 1,\\
                & \textrm{with }\widetilde{X}(t_k) \textrm{ calculated in the } (k-1) \textrm{-th iteration.}
            \end{aligned}
        \right.
    \end{align}
    Here $X(t_0)=x_0$ and $\Delta L^\alpha_k=L^\alpha_{t_{k+1}}-L^\alpha_{t_k}\overset{d}{=}\Delta t^{\frac{1}{\alpha}}L_1$ are independent with identical stable distribution. Take random variables $\boldsymbol{\xi}_k \overset{d}{=} L_1$. We can solve \eqref{NumFormLevy} by the same procedure for SDEs with additive Gaussian noise by replacing the matrix $\widetilde{A}_k$ with
    \begin{equation*}
        \widetilde{A}_k=\left(\begin{array}{cc}
            A &B \Delta t^{\frac{1}{\alpha}}\boldsymbol{\xi}_k\\
            0 &0
            \end{array} \right), \quad \widetilde{Y}(t)=\left( \begin{array}{c}
                 \widetilde{X}(t)  \\
                  \frac{1}{ \Delta t}
            \end{array}\right).
    \end{equation*}
    Here, the one stage approximation of the  approximate equation \eqref{NumFormLevy} is the same as the Euler-Maruyama scheme for the SDE \eqref{SchemeLevy}. And the Euler-Maruyama scheme itself is proven to admit strong rate of convergence in this case, see \cite{kuhn2019strong}. For $0<\alpha<1$, the moments of the $\alpha$-stable L\'evy process do not exist, hence we do not consider this case in this paper.
    In Section \ref{examp}, we present a numerical example of a one dimensional L\'evy flight. 
\end{remark}

\subsection{Linear SDEs with multiplicative Gaussian noise}

We consider the following SDE:
\begin{equation}\label{MultiSDE}
    d X(t) = A X(t) dt +  \sum_{l=1}^m B^{(l)} X(t) d W^{(l)}_t, \quad X_0=x_0.
\end{equation}
Here, $W_t=(W_t^{(1)},\cdots,W_t^{(m)})$ is a $m$-dimensional standard Brownian motion. $\{A,B^{(1)},\cdots,B^{(m)}\}$  are $d \times d$ matrices. 

We introduce the following approximate equation for the Schr\"odingerisation approach. For $t_k=kT/N_T$, we aim to compute the following ordinary differential equation,
\begin{align}\label{QM2}
    \left\{
        \begin{aligned}
            \frac{d \widetilde{X}(t)}{dt} & = \widetilde{A}_k \widetilde{X}(t) , \quad \widetilde{X}(0)=x_0 \textrm{ and } t_k< t \leq t_{k+1} \\
            \widetilde{A}_k & =A-\frac{1}{2}\sum_{l=1}^m (B^{(l)})^2+\sum_{l=1}^m B^{(l)} \frac{\Delta W^{(l)}_k}{\Delta t},
        \end{aligned}
    \right.
\end{align}
where $\Delta W^{(l)}_k=W^{(l)}_{t_{k+1}}-W^{(l)}_{t_k}$ is a standard scalar Gaussian random variable. The summation $-\frac{1}{2}\sum_{l=1}^m (B^{(l)})^2$ is from the correction term of the Ito's formula. 

Similarly as the additive noise case \eqref{TD}, we decompose $\widetilde{A}_k$ into a Hermitian matrix $H_{1,k}$ and an anti-Hermite matrix $iH_{2,k}$.
\begin{equation*}
    \widetilde{A}_k=H_{1,k}+iH_{2,k},
\end{equation*}
where 
\begin{equation*}
    H_{1,k}=\frac{\widetilde{A}_k+\widetilde{A}_k^\dag}{2}=H_{1,k}^{\dag}, \quad H_{2,k}=\frac{\widetilde{A}_k-\widetilde{A}_k^\dag}{2i}=H_{2,k}^{\dag}.
\end{equation*}
Through the warped phase transformation to $\widetilde{X}$ to get $\boldsymbol{v}(t,p)$ and the Fourier transformation on the auxiliary variable $p$ to derive $w(t,\eta)$, we arrive at the equation analogous to the Schr\"odinger equation
\begin{align}\label{MultiTH}
    \left\{
        \begin{aligned}
            \frac{d }{dt} w(t,\eta) & = -i \eta H_{1,k} w(t,\eta) + i H_{2,k} w(t,\eta), \quad t_k < t\leq t_{k+1},\\
            w(0,\eta) & =\frac{1}{\pi(1+\eta^2)} \widetilde{X}(0).
        \end{aligned}
    \right.
\end{align}
Then we can follow exactly the same steps as in the additive noise case to solve equation \eqref{MultiTH} by the Hamiltonian simulation and recover the solution through the normalized integration formula \eqref{NormIntExpli}.

Different from the additive noise case, equation \eqref{QM2} is already a homogeneous system so we do not need to expand $\widetilde{X}$ into a $(d+1)$-dimensional vector $Y$. 

In Section \ref{convrate}, we will show in Proposition \ref{Multiconvrate} the mean square norm between $\widetilde{X}$ and $X$ satisfies
\begin{equation*}
    \|\widetilde{X}(T)-X(T)\|_2 \lesssim \Delta t.
\end{equation*}

\begin{remark}
    If we assume the matrices $\{B^{(1)},\cdots,B^{(m)}\}$ commute, that is 
    \begin{equation*}
        B^{(k)}B^{(l)}=B^{(l)}B^{(k)},
    \end{equation*} 
    by the definition of the matrix exponential, we get the one stage approximation  
    \begin{align*}
        \widetilde{X}(t_{k+1})& =\widetilde{X}(t_{k})+\widetilde{A}_k \widetilde{X}(t_{k}) \Delta t+\frac{1}{2}\sum_{l=1}^m\sum_{j=1}^m B^{(l)}B^{(j)}\widetilde{X}(t_{k})\Delta W^{(l)}_k\Delta W^{(j)}_k+\mathcal{O}(\Delta t \Delta W) \\
        & = \widetilde{X}(t_{k})+A\widetilde{X}(t_{k})\Delta t+\sum_{l=1}^m B^{(l)}\widetilde{X}(t_{k}) \Delta W^{(l)}_k +\frac{1}{2}\sum_{l=1}^m (B^{(l)})^2\widetilde{X}(t_{k})((\Delta W^{(j)}_k)^2-\Delta t)\\
        & \quad  + \frac{1}{2}\sum_{\substack{l,j=1 \\l\neq j}}^m B^{(l)}B^{(j)}\widetilde{X}(t_{k})\Delta W^{(l)}_k\Delta W^{(j)}_k+\mathcal{O}(\Delta t \Delta W),
    \end{align*}
    which is the Milstein scheme(see, for example, \cite[Section 10.3]{kloeden1992}) for the SDE \eqref{MultiSDE}.

    If we further assume that matrices $\{A,B^{(1)},\cdots,B^{(m)}\}$ commute, the analytical solution to \eqref{MultiSDE} is (\cite[Section 4.8]{kloeden1992}),
    \begin{equation}\label{explisolu}
        X(t)=x_0\exp\left[(A-\frac{1}{2}\sum_{l=1}^m (B^{(l)})^2)t+\sum_{l=1}^m B^{(l)} W^{(l)}_t\right].
    \end{equation}
    In this case, it is clear that 
    \begin{equation*}
        \widetilde{X}(t_k)=x_0\exp\left[ (A-\frac{1}{2}\sum_{l=1}^m (B^{(l)})^2)t_k+ \sum_{l=1}^m B^{(l)} \sum_{j=0}^k\Delta W^{(l)}_j\right],
    \end{equation*}
    which is exactly the sample path of \eqref{explisolu}. 
\end{remark}

\subsection{Discrete numerical form}

In this section, we will numerically solve the linear differential equation arose in \eqref{TD} and \eqref{QM2}, which takes the form
\begin{equation}\label{genForm}
    \frac{d \widetilde{Z}(t)}{dt}  = \widetilde{A}_k \widetilde{Z}(t) , \quad \widetilde{Z}(0)=z_0 \textrm{ and }t_k < t \leq t_{k+1}. 
\end{equation}
Here $\widetilde{Z}=(\widetilde{X},1/\sqrt{\Delta t})$ for the additive noise case \eqref{NumForm} and $\widetilde{Z}=\widetilde{X}$ for the multiplicative noise case \eqref{QM2}. $\tilde{d}=d+1$ for the additive noise case and $\tilde{d}=d$ for the multiplicative noise case.

By the warped phase transformation, we arrive to solve the linear equation 
\begin{align}\label{Schr-1}
    \left\{
        \begin{aligned}
            \frac{d }{dt} \boldsymbol{v}(t,p) & = -H_{1,k} \partial_p \boldsymbol{v}(t,p) +iH_{2,k}\boldsymbol{v}(t,p), \quad t_k < t\leq t_{k+1},\\
            \boldsymbol{v}(0,p) & =e^{-|p|}z_0, \quad p\in(-\infty,\infty).
        \end{aligned}
    \right.
\end{align}
Next, we will numerically solve equation \eqref{Schr-1} by the discrete Fourier transform on a sufficiently large interval $[-L,L]$ so that $ e^{-L} \approx 0$. We adopt the symbols in \cite{jin2023quantum} and refer to \cite[Sec.II.A.2]{jin2023quantum} for more details.

Let $\Delta p =L/N$. For $p_j=-L+(j-1)\Delta p$, denote
\begin{equation*}
    \boldsymbol{w}(t):=[\boldsymbol{w}_1(t);\boldsymbol{w}_2(t);\cdots;\boldsymbol{w}_{2N}(t)]=\sum_{j=1}^{2N}\sum_{h=1}^{\tilde{d}}{v}_h(t,p_j)|j\rangle \otimes|h\rangle, \quad \boldsymbol{w}_j(t)=\sum_{h=1}^{\tilde{d}}{v}_h(t,p_j)|h\rangle.
\end{equation*}
Here $v_h|h\rangle$ is the $\tilde{d}$-dimensional vector with the $h$-th entry being the $h$-th entry $v_h$ of $\boldsymbol{v}$ and the other entries being zero. $|j\rangle$ is the $j$-th computational basis  corresponding to the $j$-th point in the mesh of the variable $p$. The symbol ``;" indicates the straightening of $\{\boldsymbol{w}(t,p_k)\}_{k\geq 0}$ into a column vector. 
Let $\widetilde{\boldsymbol{w}}$ be the spectral approximation of $\boldsymbol{w}$ with respect to $p$, which means 
\begin{equation*}
    \boldsymbol{w}(t)\approx \widetilde{\boldsymbol{w}}(t) = (\Phi \otimes \boldsymbol{I}) \widetilde{\boldsymbol{c}}(t).
\end{equation*}
Here, 
\begin{align*}
    & \widetilde{\boldsymbol{w}}(t):=[\widetilde{\boldsymbol{w}}_1(t);\widetilde{\boldsymbol{w}}_2(t);\cdots;\widetilde{\boldsymbol{w}}_{2N}(t)]=\sum_{j=1}^{2N}\sum_{h=1}^{\tilde{d}}\widetilde{w}_h(t,p_j)|j\rangle \otimes|h\rangle, \quad \widetilde{\boldsymbol{w}}_j(t)=\sum_{h=1}^{\tilde{d}}\widetilde{w}_h(t,p_j)|h\rangle,\\
    & \widetilde{\boldsymbol{c}}(t):=[\boldsymbol{c}_1(t);\boldsymbol{c}_2(t);\cdots;\boldsymbol{c}_{2N}(t)]=\sum_{l=1}^{2N}\sum_{h=1}^{\tilde{d}}c_l^{(h)}(t)|l\rangle \otimes|h\rangle, \quad \boldsymbol{c}_l(t)=\sum_{h=1}^{\tilde{d}}c_l^{(h)}(t)|h\rangle.
\end{align*}
and 
\begin{equation*}
    \Phi |l\rangle = \sum_{j=1}^{2N} e^{i\mu_l (p_j+L)}|j\rangle,\quad \widetilde{w}_h(t,p_j)=\sum_{l=1}^{2N} c_l^{(h)}(t)e^{i \mu_l (p_j+L)}, \quad \mu_l=\frac{\pi(l-N-1)}{L}.
\end{equation*}
It is clear that $\Phi^{-1}$ is the discrete Fourier transform on variable $p$.
Then $\widetilde{\boldsymbol{w}}$ satisfies
\begin{equation}\label{eqw}
    \partial_t \widetilde{\boldsymbol{w}}(t)=\sum_{k=0}^{N_T-1} -i(  P_{\boldsymbol{\mu}}\otimes H_{1,k} - \boldsymbol{I}\otimes H_{2,k}) \mathbbm{1}_{(t_k,t_{k+1}]}(t) \widetilde{\boldsymbol{w}}(t).
\end{equation}
Here $P_{\boldsymbol{\mu}}=-i\partial_p$ is the momentum operator and on the computational basis $\{|j\rangle\}_{1\leq j\leq 2N}$, it satisfies 
\begin{equation*}
    P_{\boldsymbol{\mu}}\Phi|l\rangle=\sum_{j=1}^{2N}\mu_l e^{i \mu_l (p_j+L)}|j\rangle.
\end{equation*}
A direct computation shows that $D_{\boldsymbol{\mu}}=\Phi^{-1} P_{\boldsymbol{\mu}} \Phi$ with $D_{\boldsymbol{\mu}}=\textrm{diag}(\mu_1,\mu_2,\cdots,\mu_{2N})$, $\mu_l=\pi(l-N-1)/L$, then we have
\begin{equation*}
    \partial_t \widetilde{\boldsymbol{c}}(t)=\sum_{k=0}^{N_T-1} -i( D_{\boldsymbol{\mu}}\otimes H_{1,k} - \boldsymbol{I} \otimes H_{2,k} ) \mathbbm{1}_{(t_k,t_{k+1}]}(t) \widetilde{\boldsymbol{c}}(t),  
\end{equation*}
which is a system of linear ordinary differential equations.
It can be rewritten as 
\begin{equation}\label{DisF}
    \partial_t \widetilde{\boldsymbol{c}}(t)=\sum_{k=0}^{N_T-1} -i (\widetilde{H}_{1,k} - \widetilde{H}_{2,k}) \mathbbm{1}_{(t_k,t_{k+1}]}(t) \widetilde{\boldsymbol{c}}(t),  
\end{equation}
with 
\begin{equation}\label{Hk}
    \widetilde{H}_{1,k}=\begin{pmatrix}
        \mu_1 H_{1,k} & & & \\
         & \mu_2 H_{1,k} & & \\
         & & \ddots & \\
         & & & \mu_{2N} H_{1,k}
    \end{pmatrix}, \quad  \widetilde{H}_{2,k}=\begin{pmatrix}
         H_{2,k} & & & \\
         &  H_{2,k} & & \\
         & & \ddots & \\
         & & &  H_{2,k}
    \end{pmatrix}.
\end{equation}
Denote $\mathcal{F}^{-1}$ as the discrete Fourier transform on variable $p$, 
\begin{equation*}
    \mathcal{F}|l\rangle=\frac{1}{\sqrt{2N}}\sum_{j=1}^{2N} e^{i\pi \frac{(l-1)(j-1)}{N}}|j\rangle, \quad \mathcal{F}^{-1}|j\rangle=\frac{1}{\sqrt{2N}}\sum_{l=1}^{2N} e^{-i\pi \frac{(l-1)(j-1)}{N}}|l\rangle.
\end{equation*}
Then $\Phi=\sqrt{2N}  S\mathcal{F}, S=\textrm{diag}([1,-1,\cdots,1,-1]_{2N \times 1})$.
Let
\begin{equation}\label{FouriPhi}
    \widetilde{\boldsymbol{c}}(0)=\frac{1}{\sqrt{2N}} ((\mathcal{F}^{-1}S) \otimes \boldsymbol{I}) \boldsymbol{w}(0).
\end{equation}
We obtain $\boldsymbol{\widetilde{w}}$ by 
\begin{equation}\label{Numw}
    \boldsymbol{\widetilde{w}}(t_k)=\sqrt{2N}  \left((S\mathcal{F})\otimes \boldsymbol{I}\right) \widetilde{\boldsymbol{c}}(t_k).
\end{equation}
Furthermore, we can also derive the discrete Fourier approximation $\boldsymbol{\widetilde{w}}^c$ of $\boldsymbol{v}$ using $\boldsymbol{\widetilde{c}}$: 
\begin{equation*}
    \boldsymbol{\widetilde{w}}^c(t_k,p)=[\widetilde{w}_1(t_k,p), \cdots, \widetilde{w}_{\tilde{d}}(t_k,p)], \quad \widetilde{w}_h(t_k,p)=\sum_{l=1}^{2N} c_l^{(h)}(t)e^{i \mu_l (p+L)}.
\end{equation*}
 The numerical approximation of $\widetilde{Z}$, denoted by $\widetilde{Z}^{\widetilde{w}}$, is recovered by the \emph{normalized integration method}
\begin{equation}\label{NormInt}
    \widetilde{Z}^{\widetilde{w}}(t_k)= \frac{ \sum_{p_j\in U_p}  \boldsymbol{\widetilde{w}}^c(t_k,p_j)\Delta p}{\sum_{p_j\in U_p}  e^{-p_j}\Delta p},
\end{equation}
where $U_p =[p^*,p^{*}+R]$. Denote the recovered solution $\widetilde{Z}^{\widetilde{w}}=(\widetilde{X}^{\widetilde{w}},\cdot)$ for the additive noise case and $\widetilde{Z}^{\widetilde{w}}=\widetilde{X}^{\widetilde{w}}$ for the multiplicative noise case. In Theorem \ref{AddThm} and Theorem \ref{MultThm}, we will show the mean square error between the solution $X$ and the recovered solution $\widetilde{X}^{\widetilde{w}}$.

\begin{remark}
    In quantum computers, we can directly obtain $\{\boldsymbol{c}(t_k)\}$ by Hamiltonian simulation. Note that equation \eqref{DisF} is the system with piecewise constant Hamiltonian and thus we can simply apply schemes for time-independent Hamiltonian simulation multiple times in this case. General time-dependent Hamiltonian simulation algorithms are also applicable here, for example, the Trotterization based methods \cite{An2021timedependent,An2022timedependent,Watkins2022Timedependent} and a newly proposed method by \cite{cao2023Timedependent} which transfers the non-autonomous system, that is the time-dependent Hamiltonian system, into an autonomous system in one-higher dimension. 

    For a verification in classical computers, we use the second order i-stable Runge-Kutta method \cite{2ndRK} to solve equation \eqref{DisF} on the temporal mesh $\{t_k=k \Delta t: \Delta t =T/N_T,k=0,1,\cdots,N_T\}$. It means 
\begin{align*}
    & \boldsymbol{\widetilde{c}}(t_{k+1}) = \boldsymbol{\widetilde{c}}(t_{k})+\Delta t K_3, \\
    & K_1 = F_k\boldsymbol{\widetilde{c}}(t_{k}), \quad K_2=F_k(\boldsymbol{\widetilde{c}}(t_{k})+\Delta t K_1/3), \quad  K_3=F_k (\boldsymbol{\widetilde{c}}(t_{k})+\Delta t K_2/2), \\
    &F_k=-i \widetilde{H}_{1,k} + i \widetilde{H}_{2,k}, \quad k=0,1,\cdots,N_T-1.
\end{align*}
 An i-stable is a scheme whose stability region contains part of the imaginary axis \cite{2ndRK}. Such a property is important since the spectral of the Schr\"odingerized differential operator are purely imaginary, thus a scheme like forward Euler method will be unstable.
\end{remark}


\subsection{Complexity}\label{complexity}

For the $d$-dimensional SDE \eqref{AddSDE}, we will need $d\times (2L)/\Delta p$ grid points for solving the approximate SDE \eqref{NumForm}. So the number of qubits $n$ that is needed for registers is
\begin{equation*}
    n\sim \textrm{log} (1/\Delta p)+\log(d).
\end{equation*}
Given a matrix or vector $H$, let $\|H\|_{\max}$ denote the largest entry of $H$ in absolute value and $s(H)$ denote the sparsity of $H$ which is the number of nonzero entries of $H$ in every column and row. Denote $\widetilde{H}_k=\widetilde{H}_{1,k}-\widetilde{H}_{2,k}$.
The following proposition is from \cite[Theorem 3]{low2017optimal} and it gives the complexity for implementing Hamiltonian simulation.
 \begin{proposition}\label{QSPcomplexity}
    A $d$-sparse Hamiltonian $\hat{H}$ on $n$ qubits with matrix elements specified to $\tau_b$ bits of precision can be simulated for time interval $t$, error $\epsilon$, and success probability at least $1-2\epsilon$ with $\mathcal{O}(td \| \hat{H}\|_{\textrm{max}}+\log(1/\epsilon)/\log\log(1/\epsilon)) $ queries and a factor $ \mathcal{O}(n+\tau_b \textrm{polylog}(\tau_b))$ additional quantum gates.
\end{proposition} 

Hereafter, we assume the precision $\tau$ and error $\epsilon$ are fixed. The following theorem gives the gate complexity of the Schr\"odingerisation approach.
\begin{theorem}\label{gatescomplex}
    For the $d$-dimensional SDE with additive Gaussian noise \eqref{AddSDE}, the Schr\"odingerisation approach to obtain $\widetilde{X}(T)$ of equation \eqref{eqw} has gate complexity
    \begin{equation*}
        N_{\text {Gate}, \text{Schr}} = \mathcal{O}\left(\frac{Td}{\Delta p}\Bigl(\log(\frac{1}{\Delta p})+\log(d)\Bigr)\Bigl(\|A\|_{\max}+|\Delta W|_{\max}/\sqrt{\Delta t}\Bigr)\right),
    \end{equation*}
    where $|\Delta W|_{\max}:= \max_{1\leq k\leq N_T}\|\Delta W_k\|_{\max}$.
\end{theorem}
\begin{proof}
    The procedure of the Schr\"odingerisation approach can be realized as follows:
\begin{equation}
    \widetilde{\boldsymbol{w}}(t_{0},\cdot) \xrightarrow{F_p^{-1}}\boldsymbol{c}(t_0) \xrightarrow{\mathrm{e}^{-i \widetilde{H}_1 \Delta t}} \boldsymbol{c}(t_{1}) \cdots \xrightarrow{\mathrm{e}^{-i \widetilde{H}_k \Delta t}} \boldsymbol{c}(t_{k}) \cdots \xrightarrow{\mathrm{e}^{-i \widetilde{H}_{N_T} \Delta t}} \boldsymbol{c}(T) \xrightarrow{F_p} \widetilde{\boldsymbol{w}}(T,\cdot)
\end{equation}
where $\widetilde{H}_k=\widetilde{H}_{1,k}-\widetilde{H}_{2,k}$ and $F_p=\Phi$ is the discrete Fourier transform. The Hamiltonian simulation using the Hamiltonian $\widetilde{H}_k$ is employed in the intermediate steps. It is clear from \eqref{FouriPhi} that $F_p$ can be realized based on quantum Fast Fourier transform (see, for example, \cite{Shor}).  So, the Schr\"odingerisation approach can be realized by quantum Fast Fourier transform and Hamiltonian simulation. According to Proposition \ref{QSPcomplexity}, we need to compute the sparsity $s(\widetilde{H}_k)$ and $\|\widetilde{H}_k\|_{\max}$. 

For  $s(\widetilde{H}_k)$, 
\begin{align*}
    s(\widetilde{H}_k) & =s(\widetilde{H}_{1,k}-\widetilde{H}_{2,k})  = s(D_{\boldsymbol{\mu}}\otimes H_{1,k} - \boldsymbol{I} \otimes H_{2,k}) \\
    & \leq \max_{1\leq l \leq 2N}s(\mu_l H_{1,k}-H_{2,k}) \leq d+1.
\end{align*}
For $\|\widetilde{H}_k\|_{\max}$, 
\begin{align*}
    \|\widetilde{H}_k\|_{\max} & =\|D_{\boldsymbol{\mu}}\otimes H_{1,k} - \boldsymbol{I} \otimes H_{2,k}\|_{\max} \\
    & \leq \max_{1 \leq l\leq 2N}\|\mu_l H_{1,k}\|_{\max}+\|H_{2,k}\|_{\max} \\
    & \leq (\pi/\Delta p+1)(\|A\|_{\max}+|\Delta W|_{\max}/\sqrt{\Delta t}).
\end{align*}
Here $|\Delta W|_{\max}:= \max_{1\leq k\leq N_T}|\Delta W_k|$.

Let $n_p$ denote the number of qubits that is needed for the $p$-variable. It known that the one-dimensional quantum Fourier transform on $n_p$ qubits can be implemented using $\mathcal{O}(n_p\log n_p)$ gates \cite{QFTgatecomplexity}. By proposition \ref{QSPcomplexity}, for fixed $m$ bits of precision and error $\epsilon$, the overall gate complexity to obtain $\widetilde{\boldsymbol{w}}(T,\cdot)$ is 
\begin{align*}
    N_{\text {Gate}, \text{Schr}} & = \mathcal{O}\big(n_p \log n_p\big) +N_T\mathcal{O}\bigl(n( \Delta t \max_{1\leq k\leq 2N}\{s(\widetilde{H}_k)\|\widetilde{H}_k\|_{\max}\})\bigr) \\
    & \leq \mathcal{O}\big(n_p \log n_p\big)+ \mathcal{O}\bigl(Td(\log(1/\Delta p)+\log(d)) (\|A\|_{\max}+|\Delta W|_{\max}/\sqrt{\Delta t})/\Delta p \bigr).\\
\end{align*}
As $n_p \sim \log(1/\Delta p)$, we have 
\begin{equation}\label{addgc}
    N_{\text {Gate}, \text{Schr}} = \mathcal{O}\left(\frac{Td}{\Delta p}\Bigl(\log(\frac{1}{\Delta p})+\log(d)\Bigr)\Bigl(\|A\|_{\max}+|\Delta W|_{\max}/\sqrt{\Delta t}\Bigr)\right).
\end{equation}
With the state $\widetilde{\boldsymbol{w}}(T, \cdot)$, to recover $\widetilde{X}(T)$, we can directly apply the projection measurement $\sum_{j^*}^{j^{**}} |j\rangle \langle j|$ where $p=j \Delta p$ and the state $\widetilde{X}(T)/\|\widetilde{X}(t)\|$ can be recovered with probability $O(\|\widetilde{X}(T)\|^2/\|\widetilde{X}(0)\|^2)$ \cite{ analog2023,Schrodingerisation1}. So, the overall gate complexity to derive $\widetilde{X}(T)$ is the same as \eqref{addgc}.
\end{proof}

With similar arguments, we can derive the gate complexity for the SDE with multiplicative noise.
\begin{corollary}
    For the $d$-dimensional SDE with multiplicative Gaussian noise \eqref{MultiSDE}, the Schr\"odingerisation approach to obtain $\widetilde{X}(T)$ has gate complexity
    \begin{equation*}
        N_{\text {Gate}, \text{Schr}} = \mathcal{O}\left(\frac{Td}{\Delta p}\Bigl(\log(\frac{1}{\Delta p})+\log(d)\Bigr)\Bigl(\|A\|_{\max}+m\|B\|^2_{\max}+ \frac{m\|B\|_{\max}|\Delta W|_{\max}}{\Delta t}\Bigr)\right),
    \end{equation*}
where $\Delta t=T/N_T$.
\end{corollary}
\begin{proof}
    Similar as Theorem \ref{gatescomplex}, we only need to compute the sparsity $s(\widetilde{H}_k)$ and $\|\widetilde{H}_k\|_{\max}$.

    As $\widetilde{A}_k$ is a $d$-dimensional matrix and 
    \begin{align*}
        \|\widetilde{A}_k\|_{\max} & = \Bigl\|A-\frac{1}{2}\sum_{l=1}^m (B^{(l)})^2+\sum_{l=1}^m B^{(l)} \frac{\Delta W^{(l)}_k}{\Delta t}\Bigr\|_{\max} \\
        & \lesssim \|A\|_{\max}+m\|B\|^2_{\max}+\frac{m\|B\|_{\max}|\Delta W|_{\max}}{\Delta t}.
    \end{align*}

    For $s(\widetilde{H}_k)$,  
    \begin{align*}
        s(\widetilde{H}_k) & =s(\widetilde{H}_{1,k}-\widetilde{H}_{2,k})  = s(D_{\boldsymbol{\mu}}\otimes H_{1,k} - \boldsymbol{I} \otimes H_{2,k}) \\
        & \leq \max_{1\leq l \leq 2N}s(\mu_l H_{1,k}-H_{2,k}) \leq d.
    \end{align*}

    For $\|\widetilde{H}_k\|_{\max}$, 
    \begin{align*}
        \|\widetilde{H}_k\|_{\max} & =\|D_{\boldsymbol{\mu}}\otimes H_{1,k} - \boldsymbol{I} \otimes H_{2,k}\|_{\max} \\
        & \leq \max_{1 \leq l\leq 2N}\|\mu_l H_{1,k}\|_{\max}+\|H_{2,k}\|_{\max} \\
        & \leq (\pi/\Delta p+1)(\|A\|_{\max}+m\|B\|^2_{\max}+\frac{m\|B\|_{\max}|\Delta W|_{\max}}{\Delta t}).
    \end{align*}
    Applying Proposition \ref{QSPcomplexity} and projection measurement $\sum_{j^*}^{j^{**}} |j\rangle \langle j|$, we derive the desired result.
\end{proof}

\begin{remark}
    Let $\tau :=\Delta t s(A) \|A\|_{\max}$. Let $\tau=\mathcal{O}(\log(1/\epsilon)/\log\log(1/\epsilon))$ and Proposition \ref{QSPcomplexity} is valid as established in \cite{low2017optimal}.
    By L\'evy modulus \cite{Levymodulus}, we know 
    \begin{equation*}
        \|\Delta W\|_{\max} = \mathcal{O}(\sqrt{\Delta t \log (1/\Delta t)}), \quad \textrm{almost surely.}
    \end{equation*}
    Then, for additive noise case, we have 
    \begin{equation*}
        \tau = \mathcal{O}(\Delta t \sqrt{\log (1/\Delta t)}/\Delta p).
    \end{equation*}
    It suffices to set $\Delta t \sqrt{\log (1/\Delta t)}/\Delta p = \mathcal{O}(\log(1/\epsilon)/\log\log(1/\epsilon))$ to meet the requirement of Proposition \ref{QSPcomplexity}.

    Similarly, for the multiplicative noise case, we have 
    \begin{equation*}
        \tau = \mathcal{O}( \sqrt{\Delta t \log (1/\Delta t)}/\Delta p).
    \end{equation*}
    Let $\sqrt{\Delta t \log (1/\Delta t)}/\Delta p = \mathcal{O}(\log(1/\epsilon)/\log\log(1/\epsilon))$ and it will meet the requirement of Proposition \ref{QSPcomplexity}.
\end{remark}

\subsection{Multi-sample Simulation}

In practice, we need to compute multiple samples of the SDE to calculate, for example, the expectation of the solution. We can also directly apply the Schr\"odingerisation approach for multiple samples simulation for SDE \eqref{AddSDE}.
Suppose we aim to simulate $N$ samples, take 
\begin{equation*}
    \widetilde{A}_{k}=\begin{pmatrix}
        A               & \cdots & 0              & B \Delta W_k^{(1)} & \cdots    & 0\\
        \vdots          & \ddots & \vdots         & \vdots             & \ddots    & \vdots\\
        0               & \cdots & A              & 0                  & \cdots    & B \Delta W_k^{(N)} \\
        \boldsymbol{0}  & \cdots & \boldsymbol{0} & \boldsymbol{0}     & \cdots    & \boldsymbol{0}
    \end{pmatrix}, \quad \widetilde{Y}=\begin{pmatrix}
        \widetilde{X}^{(1)}\\
        \vdots \\
        \widetilde{X}^{(N)}\\
        1/\Delta t \\
        \vdots \\
        1/\Delta t
    \end{pmatrix}, \quad \widetilde{Y}(0)=\begin{pmatrix}
        x_0\\
        \vdots \\
        x_0\\
        1/\Delta t \\
        \vdots \\
        1/\Delta t
    \end{pmatrix},
\end{equation*}
where $\Delta W_k^{(j)}=W_{t_{k+1}}^{(j)}-W_{t_k}^{(j)}$ is the increment of the $j$-th sample path of the standard $d$-dimensional Brownian motion and $\widetilde{X}^{(j)}$ is the $j$-th sample path of $\widetilde{X}$. Here $\widetilde{A}$ is an $N(d+1)\times N(d+1)$ matrix and $\boldsymbol{0}$ represents the $N\times d$ zero matrix. Then we can follow exactly the same procedure to solve equation \eqref{TD} with $\widetilde{A}_k$ and $\widetilde{Y}$ defined here. The sparsity of $(\widetilde{A}_k+\widetilde{A}_k^\dagger)$ is $(d+1)$ and thus the sparsity of $H_{1,k}$ and $H_{2,k}$ is $(d+1)$. So the gate complexity for $N$-samples simulation of the SDE \eqref{AddSDE} is 
\begin{equation*}
    N_{\text {Gates}, \text{Schr}} = \mathcal{O}\left(\frac{Td}{\Delta p}\Bigl(\log(\frac{1}{\Delta p})+\log(Nd)\Bigr)\Bigl(\|A\|_{\max}+\max_{1 \leq j \leq N}|\Delta W^{(j)}|_{\max}/\sqrt{\Delta t}\Bigr)\right).
\end{equation*}

Similarly, for the SDE \eqref{MultiSDE}, define
\begin{equation*}
    \widetilde{A}_{k}=\begin{pmatrix}
        \widetilde{A}_k^{(1)}    & \cdots & 0                     \\
        \vdots                   & \ddots & \vdots                \\
        0                        & \cdots & \widetilde{A}_k^{(N)} \\
        
    \end{pmatrix}, \quad \widetilde{Y}=\begin{pmatrix}
        \widetilde{X}^{(1)}\\
        \vdots \\
        \widetilde{X}^{(N)}
    \end{pmatrix}, \quad \widetilde{Y}(0)=\begin{pmatrix}
        x_0\\
        \vdots \\
        x_0
    \end{pmatrix}.
\end{equation*}
Here,
\begin{equation*}
    \widetilde{A}_k^{(j)}  =A-\frac{1}{2}\sum_{l=1}^m (B^{(l)})^2+\sum_{l=1}^m B^{(l)} \frac{\Delta W^{(j,l)}_k}{\Delta t},
\end{equation*}
and  $\Delta W_k^{(j)}=W_{t_{k+1}}^{(j)}-W_{t_k}^{(j)}$ is the increment of  the $j$-th sample path of the standard $m$-dimensional Brownian motion $W^{(j)}$ and $\Delta W_k^{(j)}=(\Delta W_k^{(j,1)},\cdots,W_k^{(j,m)})$. $\widetilde{X}^{(j)}$ is the $j$-th sample path of $\widetilde{X}$. The sparsity of $\widetilde{A}_{k}$ is $d$ and the same holds for $(\widetilde{A}_k+\widetilde{A}_k^\dagger)$. So the gate complexity for $N$-samples simulation of the SDE \eqref{MultiSDE} is 
\begin{equation*}
    N_{\text {Gate}, \text{Schr}} = \mathcal{O}\left(\frac{Td}{\Delta p}\Bigl(\log(\frac{1}{\Delta p})+\log(Nd)\Bigr)\Bigl(\|A\|_{\max}+m\|B\|^2_{\max}+ \frac{m\|B\|_{\max}\max_{1 \leq j \leq N}|\Delta W^{(j)}|_{\max} }{\Delta t}\Bigr)\right).
\end{equation*}

The computational complexity of the Euler-Maruyama scheme for simulating $N$ samples of the $d$-dimensional SDE with additive Gaussian noise \eqref{AddSDE} is $\mathcal{O}(Nd^2T/\Delta t)$,
so there is a quantum advantage when
\begin{equation*}
    \frac{Nd}{\log(Nd)} \gtrsim \frac{\Delta t}{\Delta p}\Bigl(1-\frac{\log(\Delta p)}{\log(Nd)}\Bigr)\Bigl(\|A\|_{\max}+\max_{1 \leq j \leq N}|\Delta W^{(j)}|_{\max}/\sqrt{\Delta t}\Bigr).
\end{equation*}

The computation complexity of the Milstein scheme for simulating $N$ samples of the $d$-dimensional SDE with  multiplicative Gaussian noise \eqref{MultiSDE} is $\mathcal{O}(Nmd^2T/\Delta t)$,
so there is a quantum advantage when
\begin{equation*}
    \frac{Nd}{\log(Nd)} \gtrsim \frac{\sqrt{\Delta t}}{\Delta p}\Bigl(1-\frac{\log(\Delta p)}{\log(Nd)}\Bigr)\left(\Bigl(\frac{\|A\|_{\max}}{m}+\|B\|^2_{\max}\Bigr)\sqrt{\Delta t}+\frac{\|B\|_{\max}\max_{1 \leq j \leq N}|\Delta W^{(j)}|_{\max} }{\sqrt{\Delta t}}\right).
\end{equation*}

\section{Numerical experiments}\label{examp}

In this section, we will apply the Schr\"odingerisation method to various SDEs in classical computers to illustrate the validation of the method in detail. Note that due to the no-cloning theorem, we can not compute the whole sample path in one run on the quantum computer. We present the computed sample path to illustrate the validation of our method.

\begin{example} [The Ornstein–Uhlenbeck process]

    The Ornstein–Uhlenbeck process is a type of continuous-time stochastic process that finds applications in various fields such as physics \cite{OUphysics}, biology \cite{OUbiology}, stochastic analysis \cite{gardiner2009stochastic} and etc.
    Consider a one-dimensional  Ornstein–Uhlenbeck process:
\begin{equation}\label{NumSDE}
    d X(t)=a X(t) \, dt+r\, dW_t, \quad X(0)=x_0,  
\end{equation}
where $a$ and $r$ are constants.

The explicit solution to equation \eqref{NumSDE} is 
\begin{equation*}
    X(t)=e^{at}\left(x_0+r\int_0^te^{-as}dW_t\right).
\end{equation*}
For the time mesh $\{t_j=j \Delta t: \Delta t =T/N_T,j=0,1,\cdots,N_T-1\}$, the explicit solution to the numerical formalism \eqref{NumForm} in the $k$-th iteration is 
\begin{equation}\label{ApproxNum}
    \widetilde{X}(t_{k+1})=e^{a \Delta t} \widetilde{X}(t_k)+(e^{a \Delta t}-1) \frac{r \Delta W_k}{a\Delta t}.
\end{equation}
In the following, we use $\sqrt{\Delta t} \xi_k$ to represent $\Delta W_k$, where $\xi_k \sim N(0,1)$ are independent standard Gaussian random variables. Take 
\begin{equation}
    \widetilde{A}_k=\left(\begin{array}{cc}
    a & r \xi_k\\
    0 &0
    \end{array} \right), \quad \widetilde{Y}=\left( \begin{array}{c}
         \widetilde{X}(t)  \\
          \frac{1}{ \sqrt{\Delta t}}
    \end{array}\right), \quad \boldsymbol{v}(t,p)=\left( \begin{array}{c}
        v_1(t,p)  \\
        v_2(t,p)
   \end{array}\right).
\end{equation}
The Hermitian matrices $H_{1,k}$ and $H_{2,k}$ in equation \eqref{Schr-1} then become
\begin{equation*}
    H_{1,k}=\left(\begin{array}{cc}
        a & \frac{ r}{2}\xi_k \\
        \frac{ r}{2}\xi_k & 0
    \end{array} \right), \quad 
    H_{2,k}=\left(\begin{array}{cc}
        0 & \frac{-i r}{2}\xi_k \\
        \frac{i r}{2}\xi_k & 0
    \end{array} \right).
\end{equation*} 

\begin{figure}[htbp]
    \begin{minipage}[]{1 \textwidth}
        \centerline{\includegraphics[width=15.45cm,height=9.18cm]{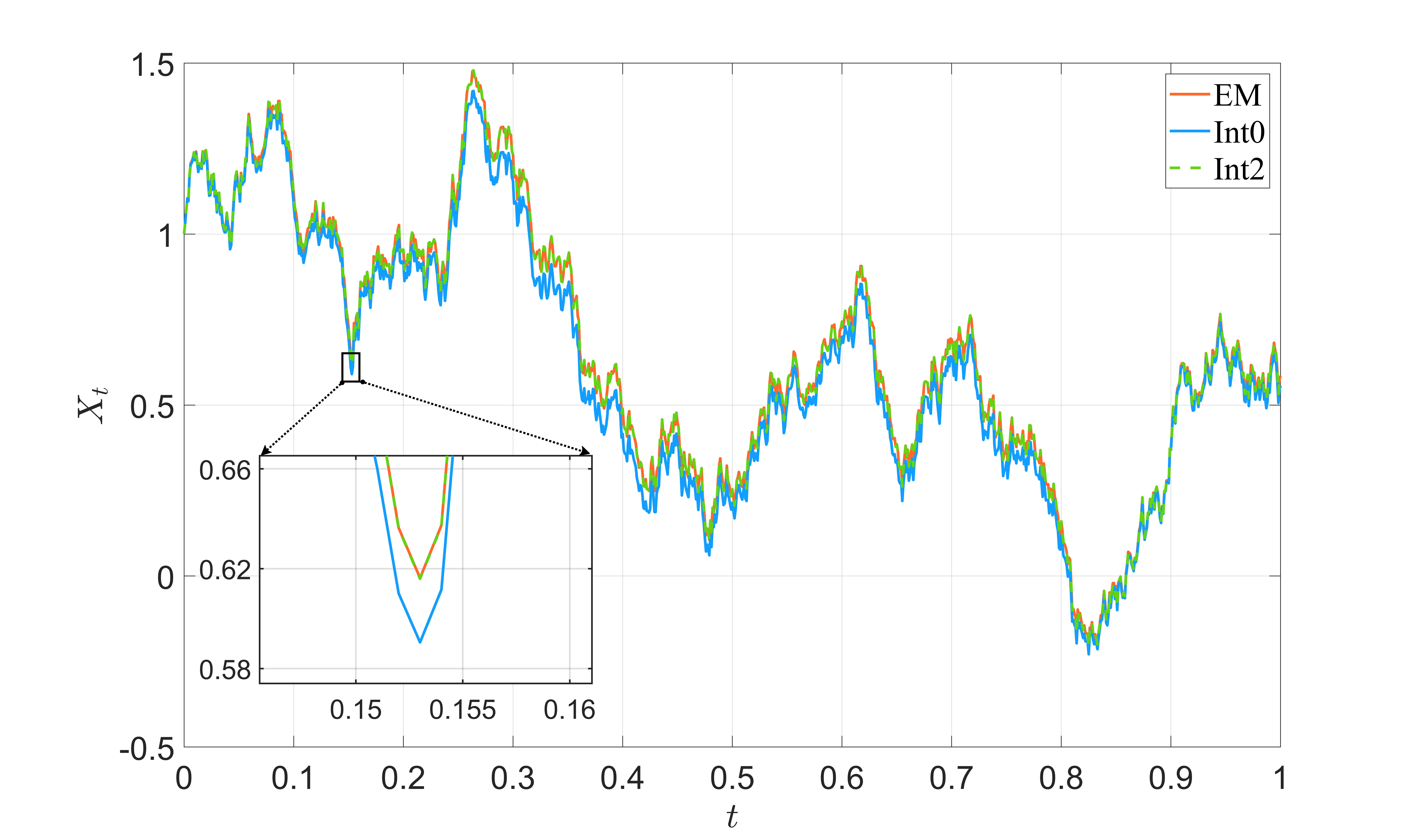}}
        \centerline{(a1)}
    \end{minipage}
    \\
    \begin{minipage}[]{0.47 \textwidth}
    \centerline{\includegraphics[width=8.62cm,height=6.3cm]{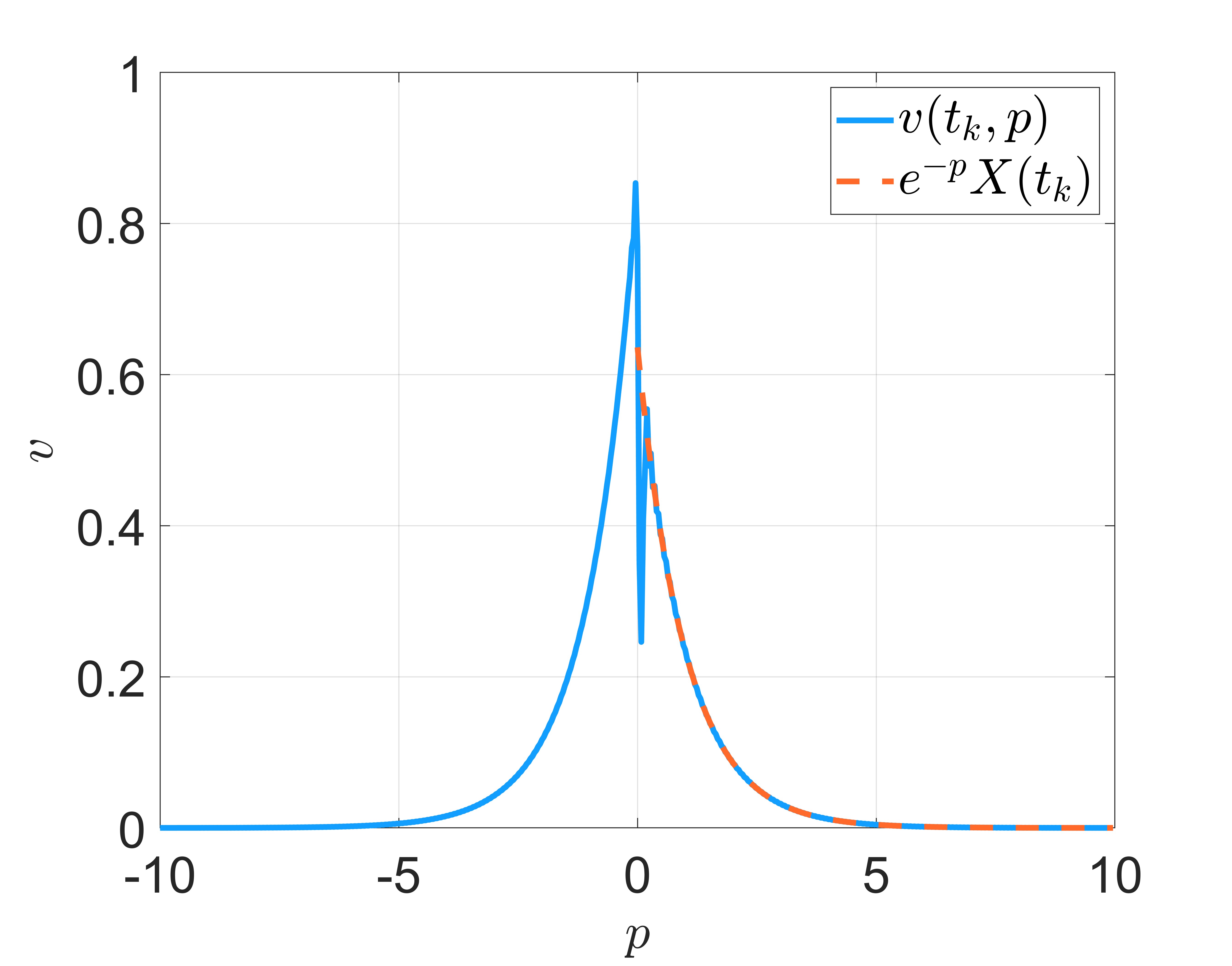}}
        \centerline{(a2)}
    \end{minipage}
    \hfill
    \begin{minipage}[]{0.47 \textwidth}
        \centerline{\includegraphics[width=8.62cm,height=6.3cm]{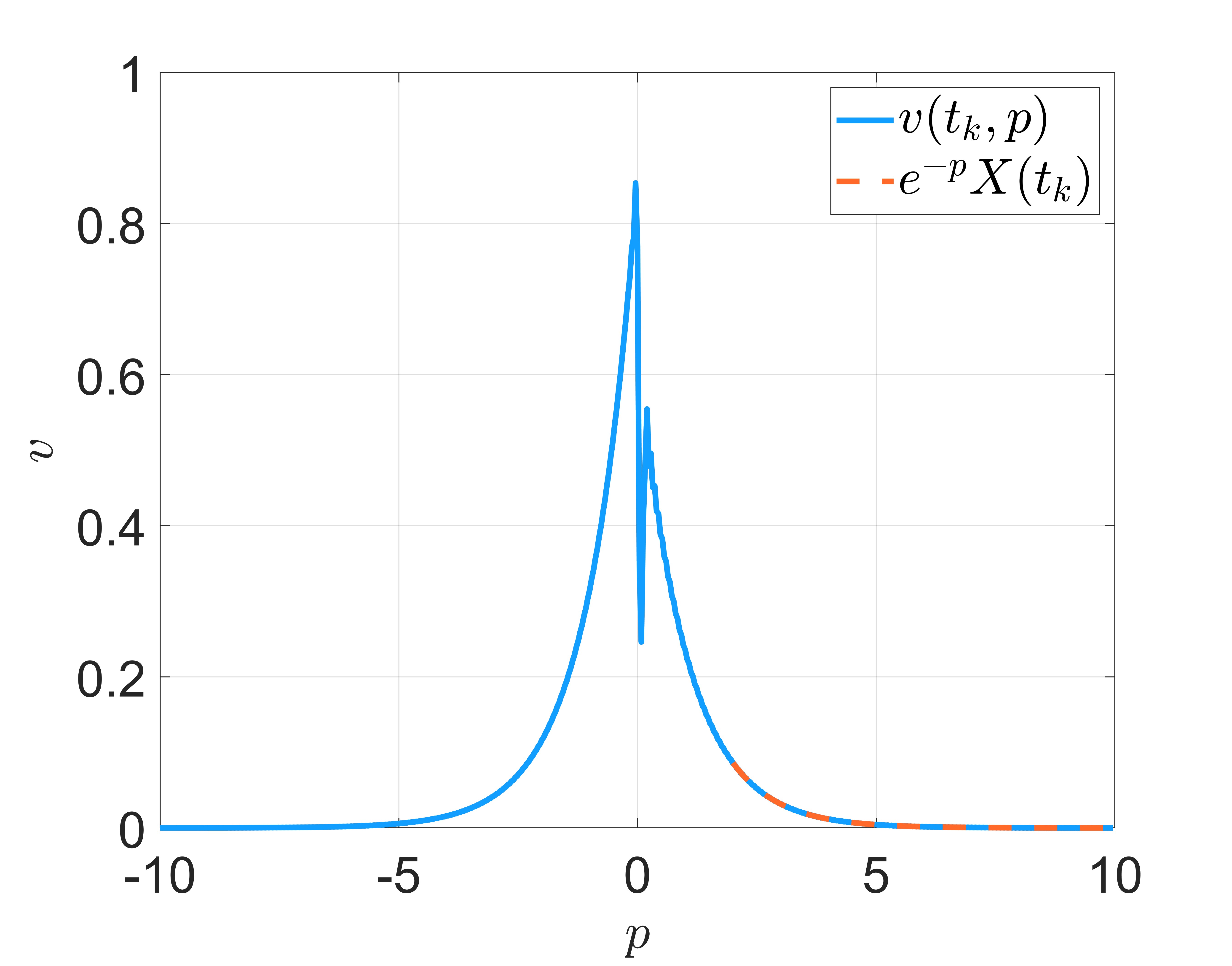}}
        \centerline{(a3)}
    \end{minipage} 

    \caption{ A sample path of equation \eqref{NumSDE} with $x_0=1$, $a=1,r=1,T=1,\Delta t=1 \times 10^{-3}$ and $\Delta p =0.04$ computed by i-stable second order Runge-Kutta method and recovered by the normalized integration method on interval $[2,10]$.  {\bf (a1)}: A sample path of the solution; {\bf (a2)}: The corresponding $v(t_k,p)$ and $e^{-p}X(t_k)$ with $t_k=0.152$ on the recovery region $[0,10]$; {\bf (a3)}: The corresponding $v(t_k,p)$ and $e^{-p}X(t_k)$ with $t_k=0.152$ on the recovery region $[2,10]$. } 
    \label{SDEIllustrationInt}
\end{figure}

For the normalized integration recovery formula \eqref{NormIntExpli} and \eqref{NormInt}, it is crucial to choose an appropriate recovery region $U$ and the corresponding $U_p$.
It is clear that with the negative diffusion term $a$ and $r=0$, the solution $v_1(t,p)$ of equation \eqref{Schr-1} will move from the right to the left in the $(p,v_1(t,p))$ plane. Because in this case the solution to equation \eqref{Schr-1} is 
\begin{equation*}
    v_1(t,p)=\exp(-|p-at|)x_0.
\end{equation*}
So, for this no-noise case, it is sufficient to choose $U=[0,\infty)$ and $U_p=[0,R]$. However, due to the nature of the noise term $r\xi_k$, the noise term can take positive values, which in turn causes the maximum eigenvalue of matrix $H_{1,k}$ become positive. According to \cite[Theorem 3.1]{Ma2024inhomov2}, the left end-point of the recovery region should be chosen to be greater than $0$, otherwise, an additional recovery error may arise. For $a>0$, even for the no-noise case, recovering the solution from the recovery region $U=[0,\infty)$ and $U_p=[0,R]$ will induce additional error. In figure \ref{SDEIllustrationInt}, we shows the influence of the recovery region for the SDE \eqref{NumSDE} with $a=1$ and $r=1$. 

For $a<0$, the largest positive eigenvalue of matrix $H_{1,k}$ is 
\begin{equation*}
    \lambda_{\textrm{max}}^+=\frac{a+\sqrt{a^2+ r^2\xi^2/4}}{2}\leq \frac{|r\xi|}{4}.
\end{equation*}
According to \cite[Theorem 3.1]{Ma2024inhomov2}, we can choose the left side of the recovery region $p^*>\lambda_{\textrm{max}}^+ T$ with $p^*=|r\xi|/4$.
By solving the ordinary differential equations \eqref{Schr-1} via second order i-stable Runge-Kutta method, we obtain the numerical solution $\boldsymbol{c}$ and $\boldsymbol{\widetilde{w}}=(\widetilde{w}_1, \widetilde{w}_2)$. Then we can obtain the recovered solution $\widetilde{X}^{\widetilde{w}}$ from $\widetilde{w}_1$ based on the normalized integration method \eqref{NormInt}. We compute $10^5$ samples and use the mean square error (MSE) which is 
\begin{equation*}
    \textrm{MSE}(\widetilde{X}^{\widetilde{w}}(T),\widetilde{X}(T)) \triangleq \sqrt{\mathbb{E}(|\widetilde{X}^{\widetilde{w}}(T)-\widetilde{X}(T)|^2)},
\end{equation*}
where $\widetilde{X}$ is computed with the same set of random variables by formula \eqref{ApproxNum}. The results are shown in Table \ref{BMMSErrorAS1T1NegaCompare}.
    
    \begin{table}
        \centering
        \begin{tabular}{ccc|c|c|c}
        \hline
        $T$ & $\Delta t$  & $\Delta p$  &\textbf{Int } & \textbf{Intp} & \textbf{EM}\\
        \hline
        1 & $2 \times 10^{-3}$ & $8 \times 10^{-2} $& $ 1.42   \times 10^{-3}$  & $ 1.07 \times 10^{-3}$ & $5.96 \times 10^{-4}$ \\
        \hline
        1 & $1 \times 10^{-3}$ & $4 \times 10^{-2} $& $ 4.26 \times 10^{-4}$  & $ 2.96 \times 10^{-4}$ & $ 2.96\times 10^{-4}$ \\
        \hline
        1 & $5 \times 10^{-4}$ & $2 \times 10^{-2} $& $ 2.47\times 10^{-4}$  & $ 1.44 \times 10^{-4}$ & $ 1.48\times 10^{-4}$ \\
        \hline
    \end{tabular}
        \caption{Mean square error  between the recovered solution $\widetilde{X}^{\widetilde{w}}$ and the approximate solution $\widetilde{X}$ with $x_0=1$, $a=-1$ and $r=1$ , using $10^5$ samples computed on a classical computer. The \textbf{Int} column is based on the normalized integration method on the fixed interval $[1.5,10]$ of the $p$-axis. The \textbf{Intp} column is based on the normalized integration method on the fixed interval $[p^*,10]$ of the $p$-axis. The \textbf{EM} column is the error between the solution $\hat{X}$ of the Euler-Maruyama scheme and the approximate solution $\widetilde{X}$. }
        \label{BMMSErrorAS1T1NegaCompare}
    \end{table}
\end{example}

\begin{example}[One dimensional L\'evy flights]
    
    L\'evy flights are a class of non-Gaussian stochastic process whose stationary increments satisfy a stationary distribution, which is heavy-tailed. L\'evy flights find applications in a wide variety of fields, including optics \cite{Levyflightsoptics}, searching \cite{Levyflightssearch}, earthquake behavior \cite{Levyflightsearthquake}, etc.

    Consider a one-dimensional L\'evy flight described by the following SDE
    \begin{equation}\label{SDELevy}
        d X(t)=\mu X(t) dt + \sigma  d L^\alpha_t, \quad X(0)=x_0,
    \end{equation}
    where $L^\alpha_t$ is a one-dimensional symmetric $\alpha$-stable L\'evy motion.
    We can follow the same steps as the case of SDEs with additive Gaussian noise, except for replacing the noise term as 
    \begin{equation*}
        \Delta L^\alpha_k\overset{d}{=}\Delta t^{\frac{1}{\alpha}} \xi_k
    \end{equation*}
    where $\xi_k \sim S(\alpha,0,0)$ is a symmetric stable random variable.
    Different from the case of the Gaussian noise, the L\'evy process contains ``large jumps" in its sample path, which means $\Delta L^\alpha_k$ may become a fairly large number and influence the eigenvalues of the Hermitian matrix $H_{1,k}$.
    In this case, a smaller $\Delta p$ will cause severe oscillation in the process of solving the ODE parts and significantly increase the error of the Schr\"odingerisation method. So we use the scale techniques in \cite{Ma2024inhomov2} to mitigate the impact off the large jumps. Define
    \begin{equation*}
        \widetilde{A}_k=\left(\begin{array}{cc}
            \mu &\epsilon\sigma \Delta t^{\frac{1}{\alpha}}\xi_k\\
            0 &0
            \end{array} \right), \quad \widetilde{Y}(t)=\left( \begin{array}{c}
                 \widetilde{X}(t)  \\
                  \frac{1}{\epsilon \Delta t}
            \end{array}\right).
    \end{equation*}
    Then, applying the Schr\"odingerisation method in Section \ref{SDEadditive}, we can compute the sample path of \eqref{SDELevy} using Hamiltonian simulation. We use the classical computer to compute $10^5$ samples and use the mean absolute error (MAE):
    \begin{equation*}
        \textrm{MAE}(X(T),Y(T))=\mathbb{E}\big|X(T)-Y(T)\big|.
    \end{equation*}
    \begin{figure}
        \begin{minipage}[]{1 \textwidth}
            \centerline{\includegraphics[width=15.45cm,height=9.18cm]{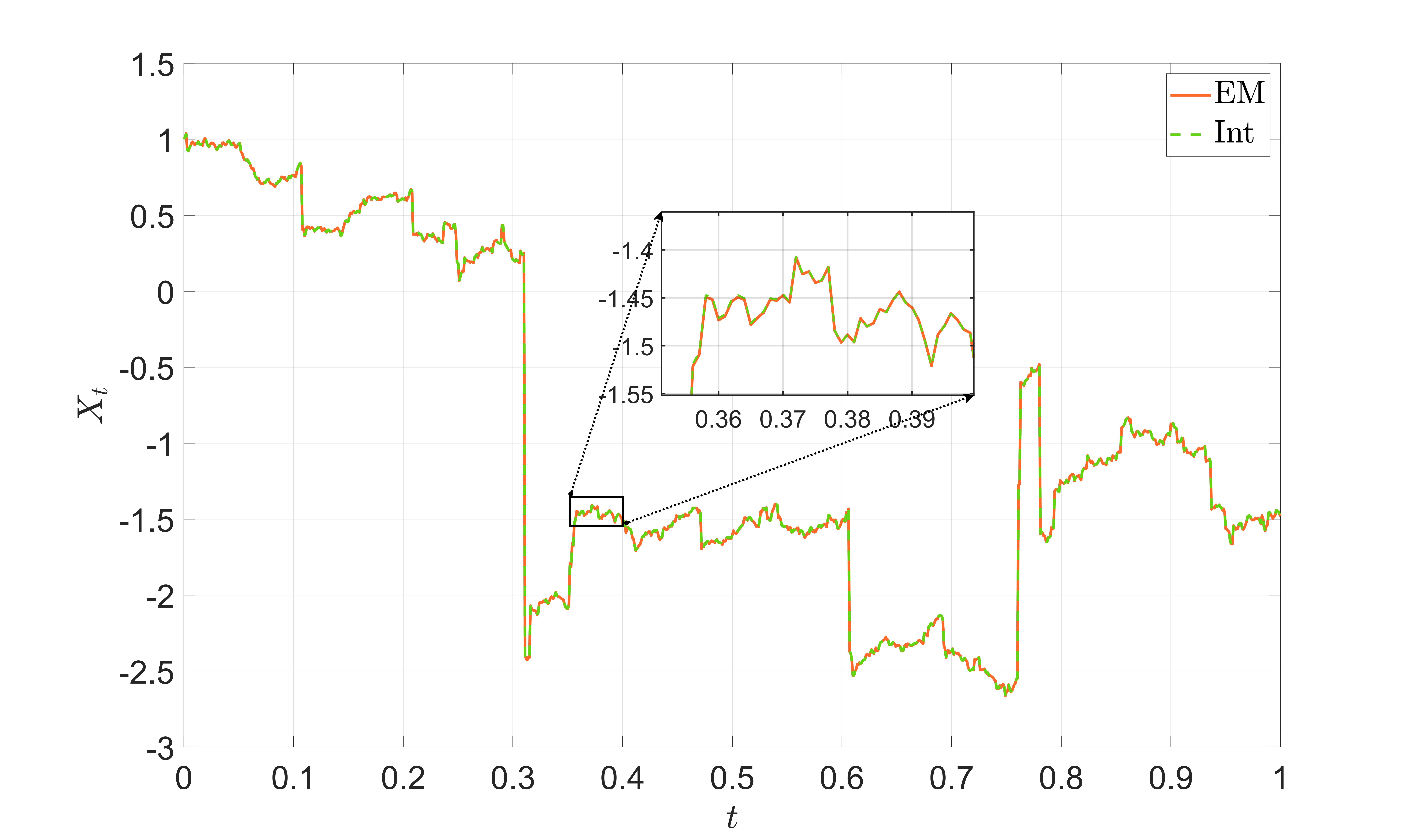}}
        \end{minipage}
        \caption{ A Sample path of equation \eqref{SDELevy} with $x_0=1$, $\mu=-1,\sigma=1,\epsilon = 0.1,T=1,\Delta t=1 \times 10^{-3}$ and $\Delta p =0.04$ computed by i-stable second order Runge-Kutta method and recovered by the normalized integration method and compared to the Euler-Maruyama solution. Large jumps take place near $t=0.3, 0.6$ and $0.76$. }
    \end{figure}
    The results are shown in Table \ref{LevyMSErrorAS1T1Nega}.
    \begin{table}[ht]
        \centering
        \begin{tabular}{ccc|c|c}
        \hline
        $T$ & $\Delta t$  & $\Delta p$  &\textbf{Int}  & \textbf{EM} \\
        \hline
        1 & $2 \times 10^{-3}$ & $8 \times 10^{-2} $& $ 2.20 \times 10^{-3}$  & $ 7.72\times 10^{-4}$  \\
        \hline
        1 & $1 \times 10^{-3}$ & $4 \times 10^{-2} $& $ 6.45\times 10^{-4}$  & $ 3.95\times 10^{-4}$  \\
        \hline
        1 & $5 \times 10^{-4}$ & $2 \times 10^{-2} $& $2.97 \times 10^{-4}$  & $1.99 \times 10^{-4}$  \\
        \hline
    \end{tabular}
        \caption{Mean absolute error compared to the approximate SDE with $x_0=1$, $\mu=-1,\sigma=1,\alpha=1.5$ and $\epsilon=0.1$, using $10^5$ samples computed on a classical computer. The \textbf{Int} column is the error between the recovered solution $\widetilde{X}^{\widetilde{w}}$ and the approximate solution $\widetilde{X}$, which is based on the normalized integration method on the fixed interval $[1.5,10]$ of the $p$-axis. The \textbf{EM} column is the error between the solution $\hat{X}$ of the Euler-Maruyama scheme and the approximate solution $\widetilde{X}$.}
        \label{LevyMSErrorAS1T1Nega}
    \end{table}
\end{example}

\begin{example}[Geometric Brownian motion]
    \begin{figure}[htbp]
        \begin{minipage}[]{1 \textwidth}
            \centerline{\includegraphics[width=15.45cm,height=9.18cm]{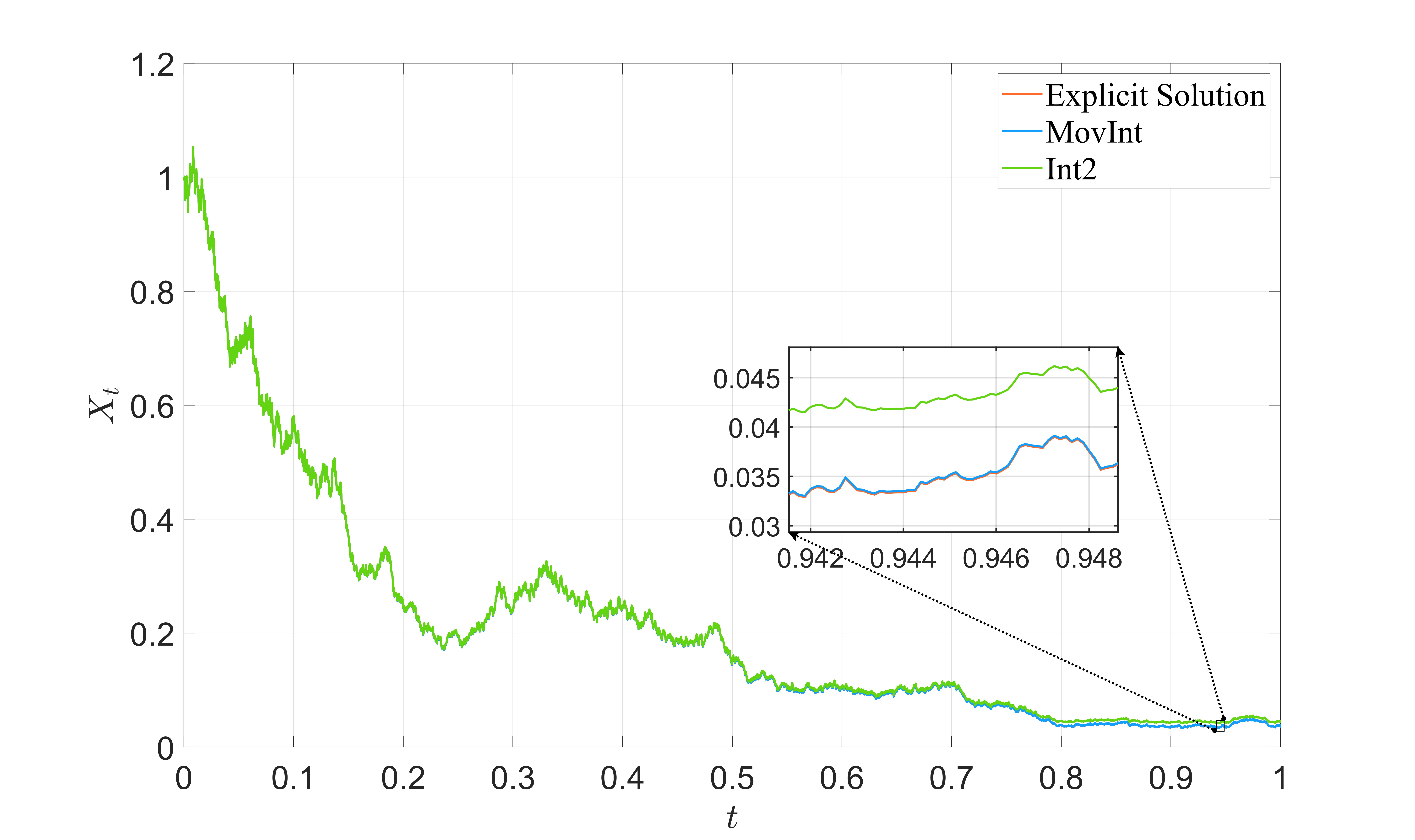}}
            \centerline{(a1)}
        \end{minipage}
        \\
        \begin{minipage}[]{0.47 \textwidth}
        \centerline{\includegraphics[width=8.62cm,height=6.76cm]{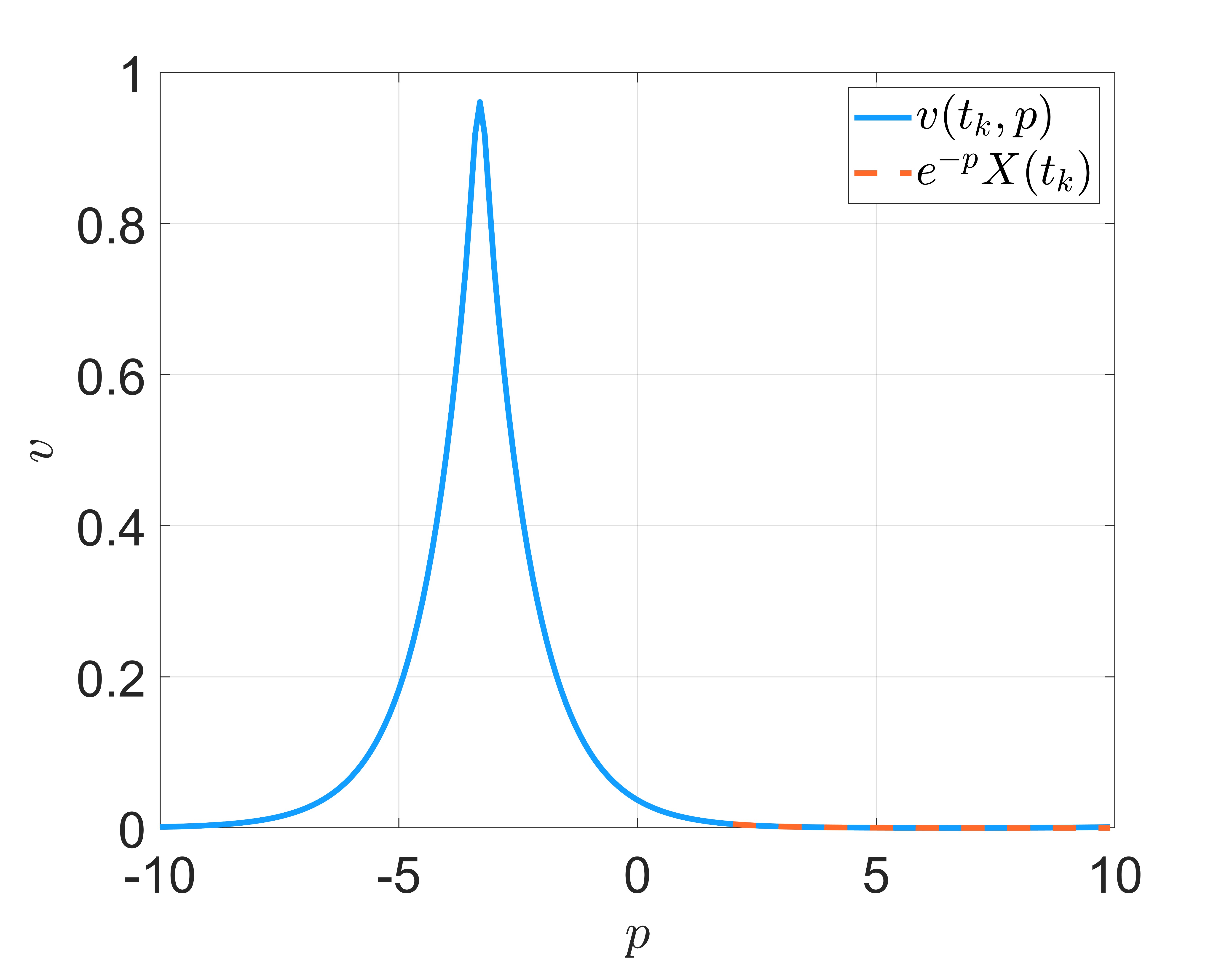}}
            \centerline{(a2)}
        \end{minipage}
        \hfill
        \begin{minipage}[]{0.47 \textwidth}
            \centerline{\includegraphics[width=8.62cm,height=6.76cm]{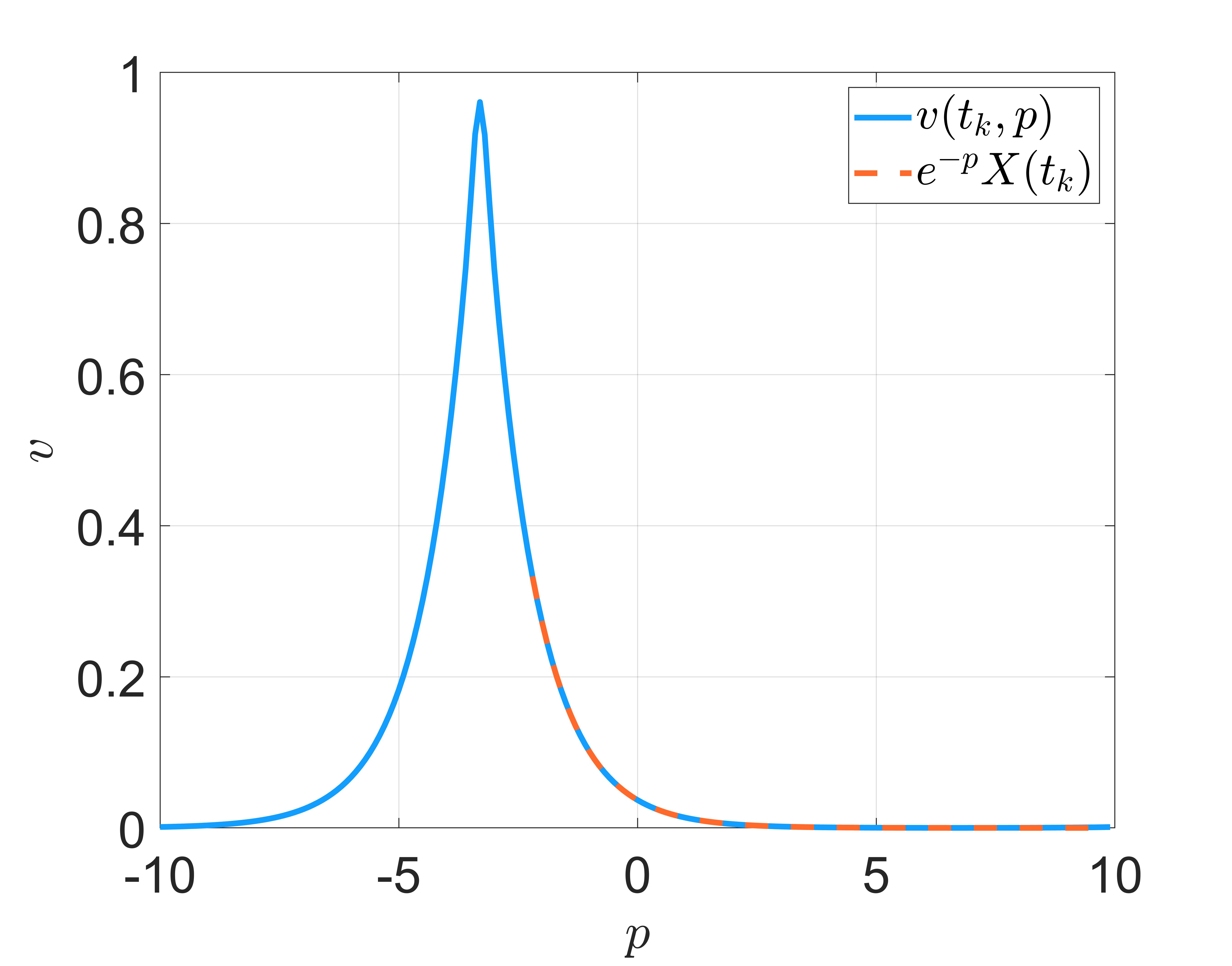}}
            \centerline{(a3)}
        \end{minipage} 
    
        \caption{ A Sample path of equation \eqref{gBM} with $x_0=1$, $\mu=-1,\sigma=1,T=1,\Delta t=1.25 \times 10^{-4}$ and $\Delta p =0.01$ computed by i-stable second order Runge-Kutta method and recovered by the normalized integration method. {\bf (a1)}: A sample path of the solution; {\bf (a2)}: The corresponding $v(t_k,p)$ and $e^{-p}X(t_k)$ with $t_k=0.9465$ on the fixed recovery region $[2,10]$; {\bf (a3)}: The corresponding $v(t_k,p)$ and $e^{-p}X(t_k)$ with $t_k=0.9465$ on the moving recovery region $[p^*(t_k)+1,10]$; }
        \label{SDEIllustrationMovInt}
    \end{figure}

    Geometric Brownian motion is a continuous time stochastic process and is known for its application in the Black-Scholes model\cite{GBMBlackScholes}.

    Consider the geometric Brownian motion described by the following one-dimensional stochastic differential equation
    \begin{equation}\label{gBM}
        d X(t)=\mu X(t) dt + \sigma X(t) d W_t, \quad X(0)=x_0.
    \end{equation}
    The explicit solution to \eqref{gBM} is 
    \begin{equation*}
        X(t)=x_0\exp\Big((\mu-\frac{1}{2}\sigma^2)t+\sigma W_t \Big).
    \end{equation*}
    For the time mesh $\{t_j=j \Delta t: \Delta t =T/N_T,j=0,1,\cdots,N_T-1\}$, the corresponding ordinary differential equation \eqref{QM2} is
    \begin{align*}
        \left\{
            \begin{aligned}
                \frac{d \widetilde{X}(t)}{dt} & = \widetilde{A}_k \widetilde{X}(t) , \quad \widetilde{X}(0)=x_0 \textrm{ and } t_k < t \leq t_{k+1}, \\
                \widetilde{A}_k & =\mu-\frac{\sigma^2}{2}+\sigma \frac{\Delta W_k}{\Delta t}.
            \end{aligned}
        \right.
    \end{align*}
    We use $\sqrt{\Delta t }\xi_k$ to represent $\Delta W_k$ and $H_{1,k}$, $H_{2,k}$ then becomes
    \begin{equation*}
        H_{1,k}=\widetilde{A}_k=\mu-\frac{\sigma^2}{2}+\sigma \frac{\xi_k}{\sqrt{\Delta t}}, \quad H_{2,k}=0.
    \end{equation*}
    We can follow exactly the same steps in the additive noise on interval $[-L,L]$ for mesh $p_j=L+(j-1)\Delta p$ with $\Delta p=(2L)/M$. It is also important to choose an appropriate recovery region $U_p$. In this example, we choose a ``moving'' $U_p$, that is, a time-dependent $U_p$.  Let 
    \begin{equation*}
        p^*(t)=\textrm{argmax}\{v(t,p_j),0\leq j \leq M\}.
    \end{equation*}
    We use the recovery region $U_p=[p^*(t)+1,R]$ to implement the normalized integration method \eqref{NormInt}.  The time-dependent recovery region $U_p$ will move with the wave $v(t,p)$ and it will further reduce the recovery error. It is because the wave $v(t,p)$ shows less oscillation near the peak $p^*$. So we can choose the left end of the recovery region near the peak. Numerical results are shown in Table \ref{MultiMSErrorAS1T1Nega} for $x_0=1$, $\mu=-1,\sigma=1$ and $T=1$, using $10^5$ samples computed on a classical computer.
    
    \begin{table}[ht]
        \centering
        \begin{tabular}{ccc|c|c|c}
        \hline
        $T$ & $\Delta t$  & $\Delta p$  &\textbf{Int2}  & \textbf{MovInt} & \textbf{EM} \\
        \hline
        1 & $5 \times 10^{-4}$ & 0.2 & $ 3.01\times 10^{-3}$ &$8.52 \times 10^{-5}$  &$ 9.51 \times 10^{-3}$  \\
        \hline
        1 & $ 2.5 \times 10^{-4}$ & 0.1& $3.31 \times 10^{-3}$ &$ 7.19 \times 10^{-5} $ & $6.77 \times 10^{-3}$  \\
        \hline
        1 & $1.25 \times 10^{-4}$ & 0.05&$ 1.20 \times 10^{-2}$ &$ 7.28\times 10^{-5}$ & $ 4.97 \times 10^{-3}$\\
        \hline
    \end{tabular}
        \caption{Mean square error compared to the explicit solution of the SDE \eqref{gBM} with $x_0=1$, $\mu=-1$ and $\sigma=1$, using $10^5$ samples computed on a classical computer. The error in the \textbf{Int2} column is based on the normalized integration method on the fixed interval $[2,10]$ of the $p$-axis. The error in the \textbf{MovInt} column is based on the normalized integration method on the moving interval $[p^*+1,10]$ of the $p$-axis. The \textbf{EM} column is the error between the solution $\hat{X}$ of the Euler-Maruyama scheme and the explicit solution $X$.}
        \label{MultiMSErrorAS1T1Nega}
    \end{table}
        
    In Table \ref{MultiMSErrorAS1T1Nega}, the error based on recovering from the fixed region $[2,10]$ (The \textbf{Int2} column) does not decrease along with $\Delta t$ and $\Delta p$. It is because the recovery region does not cover the essential part of the wave function $v(t,p)$ and it will cause an additional error that can not be improved by decreasing $\Delta t$ and $\Delta p$. However, by choosing the recovery region starting near the peak,  the performance of the scheme is significantly improved, as shown in the \textbf{MovInt} column.
    
\end{example}

\section{Analysis on convergence rate }\label{convrate}

In this section, we will give the proof of the convergence rate of the Schr\"odingerisation approach for both SDEs with additive Gaussian noise and multiplicative Gaussian noise. Specifically, we will show the mean square error between the analytical solution $X(T)$ of equation \eqref{AddSDE} and equation \eqref{MultiSDE} between their corresponding recovered solution $\widetilde{X}^{\widetilde{w}}(T)$, which collects the first $d$ components of $\widetilde{Z}^{\widetilde{w}}$ defined in \eqref{NormInt}. We know 
\begin{equation}\label{totalerror}
    \|X(T)-\widetilde{X}^{\widetilde{w}}(T)\|_2\leq \|X(T)-\widetilde{X}(T)\|_2+\sqrt{\mathbb{E}\|\widetilde{X}(T)-\widetilde{X}^{\widetilde{w}}(T)\|^2},
\end{equation}
where $\widetilde{X}$ is the analytical solution of the approximate equation \eqref{NumForm} and \eqref{QM2}. Here we use $\|\cdot\|_2=\sqrt{\mathbb{E}\|\cdot\|^2}$ to denote the mean square norm, where $\|\cdot\|$ is the Euclidean norm for vectors  and operator norm for matrices. It is clear that $\|X(T)-\widetilde{X}(T)\|_2$ is the error between the approximate equation and the SDE, and the error $\|\widetilde{X}(T)-\widetilde{X}^{\widetilde{w}}(T)\|$ comes from the recovery part.

 Denote $I=[-L,L]$, $H_1(t)=\sum_{k=0}^{N_T-1}H_{1,k}\mathbbm{1}_{[t_k,t_{k+1})}(t)$, 
\begin{align*}
    \lambda_{\max}^{+}(H_1) & =\max \{ \sup_{0<t<T} \{|\lambda|: \lambda \in \sigma(H_1(t)), \lambda >0 \}, 0\}, \\
    \lambda_{\max}^{-}(H_1) & =\max \{ \sup_{0<t<T} \{|\lambda|: \lambda \in \sigma(H_1(t)), \lambda <0 \}, 0\},
\end{align*}
 where $\sigma(H_1(t))$ collects all eigenvalues of $H_1(t)$.  
 
 We use $\|\cdot\|_{L^2(I)}$ and $\|\cdot\|_{H^1(I)}$ to denote the $L^2$-norm and $H^1$-norm on interval $I$ of variable $p$.
\begin{assumption}\label{Ass1}
    Assume the eigenvalues of $H_1$ has the following order 
    \begin{equation*}
        \lambda_1 (H_1) \leq \lambda_2 (H_1) \leq \cdots \lambda_d (H_1), \quad \textrm{for all } t \in [0,T].
    \end{equation*}
\end{assumption}

\begin{assumption}\label{Ass2}
    Let $\epsilon \ll 1$ be the desired accuracy. Assume the interval $I=[-L,L]$ satisfies 
    \begin{equation*}
        e^{-L+2\lambda_{\max}^{+}(H_1)T } \leq \epsilon, \quad e^{-L+\lambda_{\max}^{+}(H_1)T+\lambda_{\max}^{-}(H_1)T} \leq \epsilon.
    \end{equation*}
\end{assumption}

We have the following results.

\begin{theorem}\label{AddThm}
    For the $d$-dimensional SDE with $m$-dimensional additive Gaussian noise \eqref{AddSDE}, suppose the recovery region $U_p$ is properly chosen, which means the solution to \eqref{Addwptrans} satisfies $\boldsymbol{v}(t,p)=e^{-p}\widetilde{Y}(t)$ for $p\in U_p$. Suppose the Assumption \ref{Ass1} and \ref{Ass2} are satisfied.
    Then, the mean square norm between the recovered solution $\widetilde{X}^{\widetilde{w}}(T)$ and the solution $X(T)$ of the SDE \eqref{AddSDE}  satisfies 
    
    \begin{equation}
        \|X(T)-\widetilde{X}^{\widetilde{w}}(T)\|_2\lesssim C_2\Delta t + \left(\Delta p^{\frac{3}{2}}e^{\lambda_{\max}^{+}(H_1)T}+ \Delta p^{\frac{1}{2}}e^{-\widetilde{C}(T)}\right) \|\widetilde{Y}(0)\|
    \end{equation}
    where $C_2=\sqrt{{m c_1(e^{2\|A\|T}-1)}/(2\|A\|})$ with $c_1=\max_{1\leq l \leq m}|AB^{(l)}|^2$ and 
    \begin{equation*}
        \widetilde{C}(T)=\max\{L-\lambda_{\max}^{-}(H_1)T-\lambda_{\max}^{+}(H_1)T, L-2\lambda_{\max}^{+}(H_1)T \}.
    \end{equation*}
    
\end{theorem}

\begin{proof}
    Note that $\widetilde{Z} = \widetilde{Y}$ in this case. The proof follows directly from inequality \eqref{totalerror}, Lemma \ref{H1lemma}, Proposition \ref{RecProp} and Proposition \ref{Addconvrate}.
\end{proof}

Similarly, for the SDE with multiplicative Gaussian noise, we have the following result.
\begin{theorem}\label{MultThm}
    For the $d$-dimensional SDE with multiplicative Gaussian noise \eqref{MultiSDE}, suppose the recovery region $U_p$ is properly chosen, which means the solution to \eqref{Schr-1} satisfies $\boldsymbol{v}(t,p)=e^{-p}\widetilde{X}(t)$ for $p\in U_p$. Further suppose the Assumption \ref{Ass1} and \ref{Ass2} are satisfied.
    Then, the mean square norm between the recovered solution $\widetilde{X}^{\widetilde{w}}(T)$ and the solution $X(T)$ of the SDE \eqref{MultiSDE} satisfies 
    
    \begin{equation}
        \|X(T)-\widetilde{X}^{\widetilde{w}}(T)\|_2\lesssim \Delta t + \left(\Delta p^{\frac{3}{2}}e^{\lambda_{\max}^{+}(H_1)T}+ \Delta p^{\frac{1}{2}}e^{-\widetilde{C}(T)}\right) \|X(0)\|
    \end{equation}
    where 
    \begin{equation*}
        \widetilde{C}(T)=\max\{L-\lambda_{\max}^{-}(H_1)T-\lambda_{\max}^{+}(H_1)T, L-2\lambda_{\max}^{+}(H_1)T \}.
    \end{equation*}
\end{theorem}
\begin{proof}
    Note that $\widetilde{Z} = \widetilde{X}$ in this case. The proof follows directly from inequality \eqref{totalerror}, Lemma \ref{H1lemma},  Proposition \ref{RecProp} and Proposition \ref{Multiconvrate}.
\end{proof}

Following \cite{Ma2024inhomov2}, we introduce the boundary value problem:
\begin{align}\label{BVPv}
    \left\{
        \begin{aligned}
            & \frac{d }{dt} \boldsymbol{\mathcal{V}}(t,p)  = -H_{1,k} \partial_p \boldsymbol{\mathcal{V}}(t,p) +iH_{2,k}\boldsymbol{\mathcal{V}}(t,p), \quad t_k < t\leq t_{k+1},\\
            & \boldsymbol{\mathcal{V}}(t,- L)  = \boldsymbol{\mathcal{V}}(t,  L),\\
            & \boldsymbol{\mathcal{V}}(0,p)  =e^{-|p|}z_0, \quad p\in[-L,L].
        \end{aligned}
    \right.
\end{align}
Equation \eqref{BVPv} can be naturally extended to an initial value problem for $p\in (-\infty, \infty)$ via periodic extension, by setting $\boldsymbol{\mathcal{V}}(0,p)  =\mathcal{G}(p)z_0$. Here $\mathcal{G}$ is a periodic function satisfying 
\begin{equation}
    \mathcal{G}(p)=e^{-|p-2j L |} \quad (2j-1)  L \leq p <(2j+1) L.
\end{equation}

\begin{lemma}\label{H1lemma}
    For $\boldsymbol{\mathcal{V}}$ governed by  equation \eqref{BVPv}, the $H^1$-estimate of $\boldsymbol{\mathcal{V}}$ satisfies
    \begin{equation*}
        \|\boldsymbol{\mathcal{V}}(t_k,\cdot)\|_{H^1(I)}= \sqrt{2(1-e^{-2L})}\|z_0\|., \quad 0\leq k\leq N_T.
    \end{equation*}
\end{lemma}

\begin{proof}
    The Fourier series for $\boldsymbol{\mathcal{V}}(t,\cdot)$ is 
    \begin{equation*}
        \boldsymbol{\mathcal{V}}(t,p)=\sum_{l=-\infty}^{\infty} \hat{\boldsymbol{\mathcal{V}}}_l(t) e^{\frac{i l \pi }{L} p}, \quad \hat{\boldsymbol{\mathcal{V}}}_l(t) = \frac{1}{2L} \int_{-L}^{L} \boldsymbol{\mathcal{V}}(t,p)e^{-\frac{i l \pi }{L} p} dp.
    \end{equation*}
    For every $l \in \mathbb{Z}$, by equation \eqref{BVPv}, we have
    \begin{equation*}
        \frac{d}{dt} \hat{\boldsymbol{\mathcal{V}}}_l(t) = -\frac{i l \pi}{L} H_{1,k} \hat{\boldsymbol{\mathcal{V}}}_l(t) + i H_{2,k} \hat{\boldsymbol{\mathcal{V}}}_l(t), \quad t_k < t \leq t_{k+1}.
    \end{equation*}
    As matrices $H_{1,k}$ and $H_{2,k}$ are Hermite matrices, the solution operator $\exp(-i (\frac{ l \pi}{L} H_{1,k}-H_{2,k})\Delta t)$ is a unitary matrix. It means 
    \begin{equation*}
        \|\hat{\boldsymbol{\mathcal{V}}}_l(t_{k+1})\|^2 = \|\exp(-i (\frac{ l \pi}{L} H_{1,k}-H_{2,k})\Delta t)\hat{\boldsymbol{\mathcal{V}}}_l(t_{k})\|^2 = \|\hat{\boldsymbol{\mathcal{V}}}_l(t_{k})\|^2.
    \end{equation*}
    Then, 
    \begin{align*}
        \|\boldsymbol{\mathcal{V}}(t_{k+1},\cdot)\|^2_{H^1(I)} & = \sum_{l=-\infty}^{\infty} \left[1+(\frac{l\pi}{L})^2\right]\left\|\hat{\boldsymbol{\mathcal{V}}}_l(t_{k+1})\right\|^2 \\
        & = \sum_{l=-\infty}^{\infty} \left[1+(\frac{l\pi}{L})^2\right]\left\|\hat{\boldsymbol{\mathcal{V}}}_l(t_{k})\right\|^2 = \|\boldsymbol{\mathcal{V}}(t_k,\cdot)\|^2_{H^1(I)}.
    \end{align*}
    For $\|\boldsymbol{\mathcal{V}}(0,\cdot)\|_{H^1(I)}$, we know
    \begin{equation*}
        \|\boldsymbol{\mathcal{V}}(0,\cdot)\|_{H^1(\mathbb{R})}=\|z_0\|\sqrt{\int_{-L}^L e^{-2|p|}+\left|\frac{\partial}{\partial p}e^{-|p|}\right|^2 dp}=\sqrt{2(1-e^{-2L})}\|z_0\|.
    \end{equation*}

\end{proof}

For $0\leq k\leq N_T $, recall that $\boldsymbol{\widetilde{w}}^c(t_k,p)=[\widetilde{w}_1(t,p), \cdots, \widetilde{w}_{\tilde{d}}(t_k,p)]$ where $\widetilde{w}_h(t_k,p)=\sum_{l=1}^{2N} c_l^{(h)}(t)e^{i \mu_l (p+L)}$ with $c_l^{(h)}$ solved in equation \eqref{DisF}. The following lemma follows from \cite[Lemma 4.3]{Ma2024inhomov2}. It shows the $L^2$ error between $\widetilde{\boldsymbol{w}}^c$ and $\boldsymbol{v}$. 
\begin{lemma}\label{errest}
    Suppose the Assumptions \ref{Ass1} and \ref{Ass2} are satisfied.
    For the numerical solution $\widetilde{\boldsymbol{w}}^c$ defined by the solution $\widetilde{\boldsymbol{w}}$ of equation \eqref{Numw}, and the solution $\boldsymbol{v} $ of equation \eqref{Schr-1} on $[0,T]\times I$, it holds 
    \begin{equation*}
        \|\boldsymbol{v}(t_k,\cdot)-\boldsymbol{\mathcal{V}}(t_k,\cdot)\|_{L^2(I)} \lesssim (e^{- L + \lambda_{\max}^{+}(H_1)t_k }+e^{- L + \lambda_{\max}^{-}(H_1)t_k})\|z_0\|,
    \end{equation*}
    and 
    \begin{equation*}
        \|\widetilde{\boldsymbol{w} }^c(t_k,\cdot)-\boldsymbol{\mathcal{V}}(t_k,\cdot)\|_{L^2(I)} \lesssim \Delta p\|\boldsymbol{\mathcal{V}}(t_k,\cdot)\|_{H^1(I)}.
    \end{equation*}
\end{lemma}

The following proposition gives the error rate for recovered solution $\widetilde{Z}^{\widetilde{w}}$ of \eqref{NormInt}. 

\begin{proposition}[Recovery error analysis]\label{RecProp}
    Suppose the recovery region $U_p=[p^*,p^*+R)$ is properly chosen, which means the solution to \eqref{Schr-1} satisfies $\boldsymbol{v}(t,p)=e^{-p}\widetilde{Z}(t)$ for $p\in U_p$. Further assume the Assumption \ref{Ass1} and \ref{Ass2} are satisfied. 
    Then, for the recovered solution $\widetilde{Z}^{\widetilde{w}}(t_k)$ and the solution $\widetilde{Z}(t_k)$ of the approximate equation \eqref{genForm}, we have 
    
    \begin{equation}\label{rerror}
        \|\widetilde{Z}^{\widetilde{w}}(t_k)-\widetilde{Z}(t_k)\|\lesssim \left(\Delta p^{\frac{3}{2}}e^{\lambda_{\max}^{+}(H_1)T}+ \Delta p^{\frac{1}{2}}e^{-\widetilde{C}(t_k)}\right) \|\widetilde{Z}(0)\|,
    \end{equation}
    where 
    \begin{equation}
        \widetilde{C}(t_k)=\max\{L-\lambda_{\max}^{-}(H_1)t_k-\lambda_{\max}^{+}(H_1)T, L-\lambda_{\max}^{+}(H_1)t_k-\lambda_{\max}^{+}(H_1)T \}.
    \end{equation}
    
\end{proposition}

\begin{proof}
    We only need to show 
    
    \begin{equation}\label{L2estimate}
        \|\widetilde{Z}^{\widetilde{w}}(t_k)-\widetilde{Z}(t_k)\|^2 \lesssim \frac{e^{2p^*}\Delta p}{ (1-e^{-R})^2} \|\widetilde{\boldsymbol{w} }^c(t_k,\cdot)-\boldsymbol{v}(t_k,\cdot)\|^2_{L^2(I)}.
    \end{equation}
    
    According to \cite[Theorem 3.1]{Ma2024inhomov2}, it suffices to take $p^*=\lambda_{\max}^{+}(H_1)T$ and $R \gg 1$. By Lemma \ref{H1lemma} and Lemma \ref{errest}, we get the desired estimate \eqref{rerror}. 

   Now we continue to prove estimate \eqref{L2estimate}. As $\boldsymbol{v}(t,p)=e^{-p}\widetilde{Z}(t)$ for $p\in U_p$, it is clear that 
    \begin{equation*}
        \widetilde{Z}(t_k)=\frac{\int_{U_p}\boldsymbol{v}(t_k,p) dp}{\int_{U_p} e^{-p} dp}=\frac{ \sum_{p_j\in U_p}  \boldsymbol{v}(t_k,p_j)\Delta p}{\sum_{p_j\in U_p}  e^{-p_j}\Delta p}.
    \end{equation*}

    For the solution $\widetilde{Z}(t)$ of the approximate equation \eqref{genForm} and the recovered solution $\widetilde{Z}^{\widetilde{w}}(t)$ of equation \eqref{NormInt}, we have
    \begin{align*}
        \|\widetilde{Z}^{\widetilde{w}}(t_k)-\widetilde{Z}(t_k)\|^2 & = \sum_{h=1}^{\tilde{d}}|\widetilde{Z}_h^{\widetilde{w}}(t_k)-\widetilde{Z}_h(t_k)|^2 \\
        & \leq \frac{2}{\big(\sum_{p_j \in U_p}e^{-p_j}\big)^2}\sum_{h=1}^{\tilde{d}}\sum_{p_j \in U_p} \big(\widetilde{w}_h(t_k,p_j)-v_h(t_k,p_j)\big)^2.
    \end{align*}
    Choose $U_p=[p^*,p^{*}+R)=[j_1 \Delta p,j_2 \Delta p$), then
    \begin{equation*}
        \frac{1}{\sum_{p_j \in U_p}e^{-p_j}} =\frac{1}{\sum_{j=j_1}^{j_2-1}e^{-p_j}} =\frac{1-e^{-\Delta p}}{e^{-j_1 \Delta p}-e^{-j_2 \Delta p}}
    \end{equation*}
    For $\mu_l=\pi(l-N-1)/L $ and $p_j=-L+(j-1)\Delta p$, a direct computation shows
    \begin{equation*}
        \int_{-L}^{L}\left(\sum_{l=1}^{2N}\hat{u}_le^{i\mu_l (p+L)}\right)^2dp=2L\sum_{l=1}^{2N}\hat{u}_l^2=\sum_{j=1}^{2N} \left(\sum_{l=1}^{2N}\hat{u}_le^{i\mu_l( p_j+L)}\right)^2\Delta p,
    \end{equation*}
    and
    \begin{align*}
        \frac{1}{2L}\sum_{j=1}^{2N}\Bigl(\sum_{l=-\infty}^{\infty}\hat{u}_l e^{i\mu_l (p_j+L)}\Bigr)^2\Delta p & =\frac{1}{2L}\sum_{j=1}^{2N}\Bigl(\sum_{l=1}^{2N}\sum_{\tau=-\infty}^{\infty}(\hat{u}_{l+2\tau N})e^{i\mu_l (p_j+L)}\Bigr)^2\Delta p \\
        & =\sum_{l=1}^{2N}\Bigl(\sum_{\tau=-\infty}^{\infty}\hat{u}_{l+2\tau N}\Bigr)^2 \\
        & \leq 2 \sum_{l=-\infty}^{\infty} |\hat{u}_l|^2=2\Bigl\|\sum_{l=-\infty}^{\infty}\hat{u}_k e^{i\mu_l (p+L)}\Bigr\|_{L^2(I)}^2.
    \end{align*}
    Then, we have 
    \begin{align*}
        \sum_{p_j \in U_p} \big(\widetilde{w}_h(t_k,p_j)-v_h(t_k,p_j)\big)^2 \Delta p & \leq \sum_{j=1}^{2N} \big(\widetilde{w}_h(t_k,p_j)-v_h(t_k,p_j)\big)^2 \Delta p \\
        & \leq 4L \|\widetilde{w}_h(t_k,\cdot)-v_h(t_k,\cdot)\|^2_{L^2(I)}.
    \end{align*}
    Here we used the fact that 
    \begin{equation*}
        v_h(t_k,p_j)=\sum_{l=-\infty}^{\infty}\hat{v}_l^{(h)}(t_k)e^{i\mu_l(p_j+L)}, \quad \hat{v}_l^{(h)}(t_k)=\frac{1}{2L}\int_{-L}^L v_h(t_k,p)e^{-i\mu_l(p+L)}dp,
    \end{equation*}
    and 
    \begin{equation*}
        \widetilde{w}_h(t_k,p_j)=\sum_{l=1}^{2N} c_l^{(h)}(t)e^{i \mu_l (p_j+L)}=\sum_{l=-\infty}^{\infty}c_l^{(h)}(t)e^{i \mu_l (p_j+L)}, \textrm{ with } c_l^{(h)}=0 \textrm{ for } l \neq 1, \cdots, 2N.
    \end{equation*}
    We know       
    \begin{align*}
        \|\widetilde{Z}^{\widetilde{w}}(t_k)-\widetilde{Z}(t_k)\|^2 & =\sum_{h=1}^{\tilde{d}}|\widetilde{Z}_h^{\widetilde{w}}(t_k)-\widetilde{Z}_h(t_k)|^2 \\
        & \leq \frac{2}{\big(\sum_{p_j \in U_p}e^{-p_j}\big)^2}\sum_{h=1}^{\tilde{d}}\sum_{p_j \in U_p} \big(\widetilde{w}_h(t_k,p_j)-v_h(t_k,p_j)\big)^2 \\
         & \leq  \frac{8L(1-e^{-\Delta p})^2}{\Delta p(e^{-j_1 \Delta p}-e^{-j_2 \Delta p})^2}  \sum_{h=1}^{\tilde{d}}\|\widetilde{w}_h(t_k,\cdot)-v_h(t_k,\cdot)\|^2_{L^2(I)}\\ 
        & \lesssim \frac{e^{2p^*}\Delta p}{(1-e^{-R})^2} \|\widetilde{\boldsymbol{w} }^c(t_k,\cdot)-\boldsymbol{v}(t_k,\cdot)\|^2_{L^2(I)}.
    \end{align*}
    
\end{proof}

The following proposition shows the convergence rate for the one-dimension SDE:
\begin{equation}\label{1d}
    dX(t) = aX(t)dt+rdW_t, \quad X(0)=x_0.
\end{equation}
\begin{proposition}\label{Addconvproof1d}
   In one-dimension, the estimate 
   \begin{equation*}
    \|\widetilde{X}(T)-X(T)\|_2 \leq C_1 \Delta t,
   \end{equation*}
   holds for the solution $\widetilde{X}$ of the approximate equation \eqref{NumForm}, where $X$ is the solution of the SDE \eqref{1d} and $C_1={|r|\sqrt{2a(e^{2aT}-1)} }/{4}$.
\end{proposition}

\begin{proof}
    Recall that the solution to the one-dimensional SDE \eqref{1d} is 
\begin{equation*}
    X(t)=e^{at}\left(x_0+r\int_0^te^{-as}dW_s\right),
\end{equation*}
then 
\begin{equation*}
    X(t_{k+1}) =e^{a \Delta t} X(t_k)+r\int_{t_k}^{t_{k+1}} e^{a(t_{k+1}-s)}dW_s.
\end{equation*}
The solution to the corresponding approximate equation \eqref{NumForm} is 
\begin{equation*}
    \widetilde{X}(t_{k+1})=e^{a \Delta t} \widetilde{X}(t_k)+(e^{a \Delta t}-1) \frac{r \Delta W_k}{a\Delta t}=e^{a \Delta t}\widetilde{X}(t_k)+r \int_{t_k}^{t_{k+1}} \frac{e^{a\Delta t}-1}{a \Delta t} dW_s.
\end{equation*}
We have 
\begin{align*}
    \mathbb{E}\big(\widetilde{X}(t_{k+1})-X(t_{k+1})\big)^2 & = e^{2a \Delta t} \mathbb{E}\big(\widetilde{X}(t_k)-X(t_k)\big)^2+r^2\mathbb{E}\left[\int_{t_k}^{t_{k+1}} e^{a(t_{k+1}-s)}-\frac{e^{a\Delta t}-1}{a \Delta t}dW_s\right]^2 \\
    & + 2r e^{a \Delta t}\mathbb{E}\left[\big(\widetilde{X}(t_{k})-X(t_{k})\big)\int_{t_k}^{t_{k+1}} e^{a(t_{k+1}-s)}-\frac{e^{a\Delta t}-1}{a \Delta t}dW_s\right]\\
    & = e^{2a \Delta t} \mathbb{E}\big(\widetilde{X}(t_k)-X(t_k)\big)^2+r^2\int_{t_k}^{t_{k+1}} \big(e^{a(t_{k+1}-s)}-\frac{e^{a\Delta t}-1}{a \Delta t}\big)^2ds.
\end{align*}
In the last equality, we used the It\^{o} Isometry and the independence between the stochastic integral and the difference $\widetilde{X}(t_{k})-X(t_{k})$. 
For $|a|\Delta t \ll 1$ and the fact that
\begin{equation*}
    \max_{0\leq s \leq \Delta t}\left|e^{a(\Delta t-s)}-\frac{e^{a\Delta t}-1}{a \Delta t}\right|=\max\left\{\left|e^{a \Delta t}-\frac{e^{a\Delta t}-1}{a \Delta t}\right|, \left|\frac{e^{a\Delta t}-1}{a \Delta t}-1\right|\right\}=\frac{1}{2}|a| \Delta t + \mathcal{O}(\Delta t^2),
\end{equation*}
we have 
\begin{align*}
    \mathbb{E}\big(\widetilde{X}(t_{k+1})-X(t_{k+1})\big)^2 & \leq e^{2a \Delta t} \mathbb{E}\big(\widetilde{X}(t_k)-X(t_k)\big)^2+r^2 \int_0^{\Delta t} \frac{1}{4}a^2\Delta t ^2 ds+\mathcal{O}(\Delta t^3)ds \\
    & = e^{2a \Delta t} \mathbb{E}\big(\widetilde{X}(t_k)-X(t_k)\big)^2+ \frac{1}{4}(ar)^2\Delta t ^3+\mathcal{O}(\Delta t^4).
\end{align*}
Let $b_k$ denote $\mathbb{E}\big(\widetilde{X}(t_k)-X(t_k)\big)^2$, and the above inequality gives 
\begin{equation*}
    b_{k+1}\leq e^{2a \Delta t} b_k+\frac{1}{4}(ar)^2\Delta t^3+\mathcal{O}(\Delta t^4).
\end{equation*}
Then we have 
\begin{equation*}
    b_n \leq  e^{ 2a n \Delta t} b_0+\sum_{k=0}^{n-1}e^{2a k\Delta t}\big(\frac{1}{4}(ar)^2\Delta t^3+\mathcal{O}(\Delta t^4)\big).
\end{equation*}
Given $b_0=0$, which means $X(0)=\widetilde{X}(0)$, we have
\begin{equation*}
    b_n \leq \frac{(e^{2an\Delta t}-1)}{4(e^{2a\Delta t}-1)}\big((ar)^2\Delta t^3+\mathcal{O}(\Delta t^4)\big).
\end{equation*}
Take $n=k$, $t_k=k \Delta t$ and $e^{2a\Delta t}-1=2a\Delta t +\mathcal{O}((a\Delta t)^2)$, 
\begin{equation*}
    b_{k}\leq \frac{r^2a(e^{2at_k}-1) \Delta t^2}{8}+\mathcal{O}(\Delta t^3).
\end{equation*}
Then
\begin{equation*}
    \max_{t_k}\|\widetilde{X}(t_k)-X(t_k)\|_2\leq \frac{|r|\sqrt{2a(e^{2aT}-1)} \Delta t}{4}+\mathcal{O}(\Delta t^{\frac{3}{2}}),
\end{equation*}
and 
\begin{equation*}
    \|\widetilde{X}(T)-X(T)\|_2 \leq \frac{|r|\sqrt{2a(e^{2aT}-1)} \Delta t}{4}+\mathcal{O}(\Delta t^{\frac{3}{2}}).
\end{equation*}
It means that the solution of the approximate equation \eqref{NumForm} converges to the solution of equation \eqref{NumSDE} with error rate $\mathcal{O}(\Delta t)$ under mean square norm.
\end{proof}

We will next show the convergence rate for the $d$-dimensional case.

Consider the $d$-dimensional SDE
\begin{equation}\label{dSDE}
    d X(t) =A X(t)\, dt+B\, dW_t= A X(t)\, dt + \sum_{l=1}^m B^{(l)} dW^{(l)}_t, \quad X_0=x_0,
\end{equation}
where $A$ is a $d \times d $ matrix, $\{B^{(l)}\}$ is a group of  $d \times 1$ vector and $\{W^{(l)}\}$ is a group of independent scalar standard Brownian motion. Here, we use the notation that $B=(B^{(1)},\cdots, B^{(m)})$ and $W_t=(W^{(1)},\cdots,W^{(m)})^\top$ and $m$ scales comparably to $d$.

\begin{proposition}\label{Addconvrate}
    For the $d$-dimensional SDE with additive Gaussian noise \eqref{dSDE}, the estimate 
    \begin{equation*}
        \|\widetilde{X}(T)-X(T)\|_2 \leq C_2 \Delta t,
    \end{equation*}
    holds for the solution $\widetilde{X}$ of the approximate equation \eqref{NumForm}, where $C_2=\sqrt{{m c_1(e^{2\|A\|T}-1)}/(2\|A\|})$ and $c_1=\max_{1\leq l \leq m}|AB^{(l)}|^2$.
\end{proposition}
\begin{proof}
    According to \cite[Ch 4.8]{kloeden1992}, the solution to \eqref{dSDE} is 
\begin{equation*}
    X(t)=e^{At}x_0+\int_0^t e^{A(t-s)}B\,dW_s 
\end{equation*}
Then 
\begin{align*}
    X(t_{k+1}) & =e^{A(t_k+\Delta t)}x_0+\int_0^{t_k}e^{A(t_k+\Delta t - s)}BdW_s+\int_{t_k}^{t_{k+1}}e^{A(t_{k+1}-s )}B dW_s \\
    & = e^{A\Delta t}X(t_k)+ \int_{t_k}^{t_{k+1}}e^{A(t_{k+1}-s )}B dW_s= e^{A\Delta t}X(t_k)+ \sum_{l=1}^{m}\int_{t_k}^{t_{k+1}}e^{A(t_{k+1}-s )}B^{(l)} dW_s^{(l)}.
\end{align*}
For the corresponding approximate equation \eqref{NumForm}, by the Duhamel principle, we obtain
\begin{equation*}
    \widetilde{X}(t_{k+1})=e^{A \Delta t} \widetilde{X}(t_k)+\int_{t_k}^{t_{k+1}} e^{A(t_{k+1}-s)}\frac{B \Delta W_k}{\Delta t} ds =e^{A \Delta t} \widetilde{X}(t_k)+\sum_{l=1}^m\int_{t_k}^{t_{k+1}} e^{A(t_{k+1}-s)}\frac{B^{(l)} \Delta W_k^{(l)}}{\Delta t}ds.
\end{equation*}
Therefore,
\begin{align*}
    & \mathbb{E}\big\|\widetilde{X}(t_{k+1})-X(t_{k+1})\big\|^2 \\
    & = \mathbb{E}\big\|e^{A\Delta t}(\widetilde{X}(t_k)-X(t_k))\big\|^2+\sum_{l=1}^m \mathbb{E}\left\|\int_{t_k}^{t_{k+1}}\int_0^{\Delta t}\frac{e^{A(\Delta t-s)}B^{(l)}}{\Delta t}-\frac{e^{A(t_{k+1}-\tau)}B^{(l)}}{\Delta t}ds dW_\tau \right\|^2 \\
    & + 2\mathbb{E}\Big[\Big(e^{A\Delta t}\big(\widetilde{X}(t_k)-X(t_k)\big)\Big)^\dagger \sum_{l=1}^m \int_{t_k}^{t_{k+1}}\int_0^{\Delta t}\frac{e^{A(\Delta t-s)}B^{(l)}}{\Delta t}-\frac{e^{A(t_{k+1}-\tau)}B^{(l)}}{\Delta t}ds dW_\tau \Big]\\
    & \leq e^{2\|A\|\Delta t}\mathbb{E}\big\|\widetilde{X}(t_k)-X(t_k)\big\|^2+ \sum_{l=1}^m \int_0^{\Delta t}\left\|\int_0^{\Delta t}\frac{e^{A(\Delta t-s)}B^{(l)}}{\Delta t}-\frac{e^{A(\Delta t-\tau)}B^{(l)}}{\Delta t}ds \right\|^2d\tau.
\end{align*}
In the last equality, we used the It\^{o} Isometry and the independence between the stochastic integral and the difference $\widetilde{X}(t_{k})-X(t_{k})$. 
 By the definition of the matrix exponential, for $\tau \in [0,\Delta t]$ and $\|A\|\Delta t \ll 1 $, we have
\begin{align*}
    & \left\|\int_0^{\Delta t}\frac{e^{A(\Delta t-s)}B^{(l)}}{\Delta t}-\frac{e^{A(\Delta t-\tau)}B^{(l)}}{\Delta t}ds \right\|^2\\
    & = \frac{1}{\Delta t^2}\left\|\int_0^{\Delta t}\Big(\textrm{I}_d+A(\Delta t-s)+\frac{1}{2}A^2(\Delta t - s)^2+\mathcal{O}(\Delta t^3)\Big)B^{(l)} \right.\\
    & \left.-\Big(\textrm{I}_d+A(\Delta t-\tau)+\frac{1}{2}A^2(\Delta t - \tau)^2+\mathcal{O}(\Delta t^3)\Big)B^{(l)}  ds\right\|^2\\
    & =\frac{1}{\Delta t^2}\left\|\int_0^{\Delta t}\Big(A(\tau-s)+\frac{1}{2}A^2(\tau-s )(2 \Delta t - \tau -s)+\mathcal{O}(\Delta t^3)\Big)B^{(l)} ds\right\|^2 \\
    & = \frac{1}{\Delta t^2}\left\|AB^{(l)}\int_0^{\Delta t}(\tau-s)ds+A^2B^{(l)}\int_0^{\Delta t}(\tau-s )(2 \Delta t - \tau -s)ds+\mathcal{O}(\Delta t^4)\right\|^2 \\
    & \leq \frac{1}{2}\|AB^{(l)}\|^2\Delta t^2+\frac{2}{9}\|A^2B^{(l)}\|^2\Delta t^4+\mathcal{O}(\Delta t^6).
\end{align*}
Then, take $c_1=\max_{1\leq l \leq m}\frac{1}{2}\|AB^{(l)}\|^2$ and $c_2=\max_{1\leq l \leq m}\frac{2}{9}\|A^2B^{(l)}\|^2$, we have 
\begin{align*}
    \mathbb{E}\big\|\widetilde{X}(t_{k+1})-X(t_{k+1})\big\|^2 & \leq  e^{2\|A\|\Delta t}\mathbb{E}\big\|\widetilde{X}(t_k)-X(t_k)\big\|^2+m\big(c_1\Delta t^3+c_2\Delta t^5+\mathcal{O}(\Delta t^7)\big).
\end{align*}
Similarly as the one-dimensional case, take $b_k=\mathbb{E}\big\|\widetilde{X}(t_k)-X(t_k)\big\|^2$, 
\begin{align*}
    b_{k+1} & \leq e^{ 2\|A\|\Delta t}b_k+m c_1\Delta t^3+\mathcal{O}(\Delta t^5) \\
        & = e^{2 ({k+1})\|A\|\Delta t}b_0+m\sum_{j=0}^{k}e^{2j\|A\|\Delta t }\big(c_1\Delta t^3+\mathcal{O}(\Delta t^5)\big) \\
        & = \frac{mc_1(e^{2(k+1)\|A\|\Delta t}-1)}{e^{2\|A\|\Delta t}-1}\big(\Delta t^3+\mathcal{O}(\Delta t^5)\big).
\end{align*}
Take $t_k=k \Delta t$ and $e^{2\|A\|\Delta t}-1=2\|A\|\Delta t +\mathcal{O}((\|A\|\Delta t)^2)$, we have
\begin{equation*}
    b_{k}\leq  \frac{m c_1(e^{2\|A\|t_k}-1)}{2\|A\|} \Delta t^2+\mathcal{O}(\Delta t^4).
\end{equation*}
It means 
\begin{equation*}
    \max_{t_k}\|\widetilde{X}(t_k)-X(t_k)\|_2 \leq \sqrt{\frac{m c_1(e^{2\|A\|T}-1)}{2\|A\|}} \Delta t+\mathcal{O}(\Delta t^2),
\end{equation*}
and 
\begin{equation*}
     \|\widetilde{X}(T)-X(T)\|_2 \leq \sqrt{\frac{m c_1(e^{2\|A\|T}-1)}{2\|A\|}} \Delta t+\mathcal{O}(\Delta t^2).
\end{equation*}
In $d$-dimensions, the solution of the approximate equation \eqref{NumForm} converges to the solution of equation \eqref{NumSDE} with error rate $\mathcal{O}(\sqrt{m}\Delta t)$ under mean square norm.
\end{proof}

In the following, we will move on to show the convergence rate for SDEs with multiplicative Gaussian noise.

Recall the SDE with multiplicative Gaussian noise:
\begin{equation}\label{MultiProof}
    d X(t) = A X(t) dt +  \sum_{l=1}^m B^{(l)} X(t) d W^{(l)}_t, \quad X_0=x_0,
\end{equation}
and the approximate equation, 
\begin{align}\label{MultiApproProof}
    \left\{
        \begin{aligned}
            & \frac{d \widetilde{X}(t)}{dt}  = \widetilde{A}_k \widetilde{X}(t) , \quad \widetilde{X}(0)=x_0 \textrm{ and } t_k < t \leq t_{k+1}, \\
            & \widetilde{A}_k  =A-\frac{1}{2}\sum_{l=1}^m (B^{(l)})^2+\sum_{l=1}^m B^{(l)} \frac{\Delta W^{(l)}_k}{\Delta t}.
        \end{aligned}
    \right.
\end{align}
Here, $W_t=(W_t^{(1)},\cdots,W_t^{(m)})$ is a $m$-dimensional standard Brownian motion and $\Delta W_k^{(l)}=W_{t_{k+1}}^{(l)}-W_{t_{k}}^{(l)}$. $\{A,B^{(1)},\cdots,B^{(m)}\}$ are $d \times d$ matrices, $\Delta t=1/N_T$ and $t_k=kT/N_T$ for $k=0,1,\cdots,N_T$.
\begin{proposition}\label{Multiconvrate}
    For the $d$-dimensional SDE with multiplicative Gaussian noise\eqref{MultiProof}, the estimate
    \begin{equation*}
        \|\widetilde{X}(T)-X(T)\|_2 \lesssim  \Delta t,
    \end{equation*}
    holds for the solution $\widetilde{X}$ of the approximate equation\eqref{MultiApproProof}. 
\end{proposition}

\begin{proof}
    We show this estimate by comparing the solution $\widetilde{X}$ to the Milstein approximation $\widehat{X}$:
    \begin{equation}\label{MLscheme}
        \widehat{X}(t_{k+1})=\widehat{X}(t_k)+A \widehat{X}(t_k) \Delta t+\sum_{l=1}^m B^{(l)} \widehat{X}(t_k) \Delta W_k^{(l)}+ \sum_{l,j =1}^m B^{(l)}B^{(j)}\widehat{X}(t_k)I_k(j,l), \quad \widehat{X}_{t_0}=x_0.    
    \end{equation}
    Here $I_k(j,l)$ represents the Ito's integral
    \begin{equation*}
        I_k(j,l) = \int_{t_k}^{t_{k+1}}(W_{s}^{(j)}-W_{t_k}^{(j)})dW_s^{(l)}.
    \end{equation*}
    Using Ito's formula, it is clear that
    \begin{equation*}
        I_k(l,j)+I_k(j,l)=\Delta W_k^{(l)}\Delta W_k^{(j)} \quad I_k(l,l)= \frac{1}{2}\bigl((\Delta W_k^{(l)})^2-\Delta t\bigr).
    \end{equation*}
    In the $(k+1)$-th iteration of \eqref{MultiApproProof}, that is obtaining $\widetilde{X}(t_{k+1})$ from $\widetilde{X}(t_k)$, we know 
    \begin{equation*}
        \widetilde{X}(t_{k+1})=\exp(\widetilde{A}_k \Delta t)\widetilde{X}(t_k).
    \end{equation*}
        By the definition of the matrix exponential, we have
    \begin{equation*}
        \exp(\widetilde{A}_k \Delta t)=I+\widetilde{A}_k \Delta t+\frac{1}{2}\sum_{l,j=1}^m B^{(l)}B^{(j)}\Delta W^{(l)}_k\Delta W^{(j)}_k+\mathcal{O}(\Delta t \Delta W).
    \end{equation*}
    Comparing to the Milstein approximation $\widehat{X}$ in \eqref{MLscheme}, we obtain
    \begin{align}\label{MLCompare}
        \widehat{X}(t_{k+1})- & \widetilde{X}(t_{k+1}) =\widehat{X}(t_k)-\widetilde{X}(t_k)+A\big(\widehat{X}(t_k)-\widetilde{X}(t_k)\big)\Delta t +\sum_{l=1}^m B^{(l)}\big(\widehat{X}(t_k)-\widetilde{X}(t_k)\big)\Delta W^{(l)}_k\\
        & +\frac{1}{2}\sum_{l=1}^m(B^{(l)})^2(\widehat{X}(t_k)-\widetilde{X}(t_k))\big((\Delta W_k^{(l)})^2-\Delta t\big)+\frac{1}{2}\sum_{\substack{l,j=1 \\l\neq j}}^m B^{(l)}B^{(j)}(\widehat{X}(t_k)-\widetilde{X}(t_k))I_k(j,l)\nonumber \\ 
        & +\frac{1}{2}\sum_{\substack{l,j=1 \\l\neq j}}^m B^{(l)}B^{(j)}\widetilde{X}(t_k)(I_k(j,l)-\Delta W_k^{(l)}\Delta W_k^{(j)}) +\mathcal{O}(\Delta t \Delta W). \nonumber 
    \end{align}
It is clear that 
\begin{equation*}
    \left\|A\big(\widehat{X}(t_k)-\widetilde{X}(t_k)\big)\right\|\leq \|A\|\left\|\widehat{X}(t_k)-\widetilde{X}(t_k)\right\|, \quad \left\|B^{(l)}\big(\widehat{X}(t_k)-\widetilde{X}(t_k)\big)\right\|\leq \|B^{(l)}\|\left\|\widehat{X}(t_k)-\widetilde{X}(t_k)\right\|.
\end{equation*}
We will next show that $\mathbb{E}\|\widetilde{X}(t_{k})\|^2$ and $\mathbb{E}\|\widetilde{X}(t_{k})\|$ is bounded for all $k$. The boundedness of $\mathbb{E}\|\widetilde{X}(t_{k})\|$ follows from $\mathbb{E}\|\widetilde{X}(t_{k})\|^2$ by H\"older's inequalities. A direct computation gives
\begin{align*}
    \mathbb{E}\|\widetilde{X}(t_k)\|^2 & =\mathbb{E}\|\exp(\widetilde{A}_{k-1}\Delta t)\exp(\widetilde{A}_{k-2}\Delta t)\cdots \exp(\widetilde{A}_0\Delta t)\widetilde{X}(0)\|^2 \\
    & \leq \mathbb{E}\left(\exp(2\|\widetilde{A}_{k-1}\|\Delta t) \exp(2\|\widetilde{A}_{k-2}\|\Delta t)\cdots \exp(2\|\widetilde{A}_{0}\|\Delta t)\|\widetilde{X}(0)\|^2\right) \\
    & = \mathbb{E}\Bigl(\|\widetilde{X}(0)\|^2\exp(2\sum_{\tau=0}^{k-1}\|\widetilde{A}_{\tau}\|\Delta t)\Bigr)
\end{align*}
By the definition of $\widetilde{A}_k$ and the fact that $\Delta W_k^{(l)}$ are independent with each other, we have for every $k$,
\begin{equation}
    \mathbb{E}\|\widetilde{X}(t_k)\|^2  \leq \|x_0\|^2\exp\Bigr[2(\|A\|+\frac{1}{2}\sum_{l=1}^m\|B^{(l)}\|^2)T\Bigl]\prod_{\tau=0}^{N_T-1}\prod_{l=1}^m\mathbb{E}\Bigr[\exp(2\|B^{(l)}\||\Delta W_k^{(l)}|)\Bigl]<\infty.
\end{equation}
Squaring both sides and taking expectation of \eqref{MLCompare} gives
\begin{align*}
    & \mathbb{E}\big\|\widehat{X}(t_{k+1}) -\widetilde{X}(t_{k+1})\big\|^2 = \mathbb{E}\big[(\widehat{X}(t_{k+1}) -\widetilde{X}(t_{k+1}))^\dagger(\widehat{X}(t_{k+1}) -\widetilde{X}(t_{k+1}))\big]\\
    & =\mathbb{E}\big\|\widehat{X}(t_k)-\widetilde{X}(t_k)\big\|^2+\mathbb{E}\Big\|A\big(\widehat{X}(t_k)-\widetilde{X}(t_k)\big)\Big\|^2\Delta t^2 +\sum_{l=1}^m \mathbb{E}\Big\|B^{(l)}\big(\widehat{X}(t_k)-\widetilde{X}(t_k)\big)\Big\|^2\Delta t \\
    & +\frac{1}{4}\mathbb{E}\Big\|\sum_{l=1}^m(B^{(l)})^2\big(\widehat{X}(t_k)-\widetilde{X}(t_k)\big)\big((\Delta W_k^{(l)})^2-\Delta t\big)\Big\|^2 +\frac{1}{4}\mathbb{E}\Big\| \sum_{\substack{l,j=1 \\l\neq j}}^m B^{(l)}B^{(j)}\big(\widehat{X}(t_k)-\widetilde{X}(t_k)\big)I_k(j,l)\Big\|^2\\
    & +\frac{1}{4}\mathbb{E}\Big\| \sum_{\substack{l,j=1 \\l\neq j}}^m B^{(l)}B^{(j)}\widetilde{X}(t_k)(I_k(j,l)-\Delta W_k^{(l)}\Delta W_k^{(j)})\Big\|^2+\mathbb{E}\Big(\big(\widehat{X}(t_{k})-\widetilde{X}(t_{k}\big)^\dagger A\big(\widehat{X}(t_k)-\widetilde{X}(t_k)\big)\Delta t\Big) \\
    & -\frac{1}{4}\mathbb{E}\Big(\sum_{\substack{l,j=1 \\l\neq j}}^m\big(R_k^{(l,j)}+(R_k^{(l,j)})^\dagger\big)(I_k(j,l))^2\Big)+\mathcal{O}(\Delta t^3). 
\end{align*}
Here, we denote
\begin{equation*}
    R_k^{(l,j)} \triangleq \big(\widehat{X}(t_{k})-\widetilde{X}(t_{k})\big)^\dagger\big(B^{(j)}B^{(l)}\big)^\dagger B^{(l)}B^{(j)} \widetilde{X}(t_{k}).
\end{equation*}
and use the fact that 
\begin{equation*}
    \mathbb{E}(I_k(j_1,l_1)I_k(j_2,l_2))=\frac{1}{2}\delta_{j_1 j_2}\delta_{l_1 l_2}(\Delta t)^2, \quad \mathbb{E}(I_k(j_1,l_1)\Delta W_k^{(l_2)})=0.
\end{equation*}
By the fact that $\widehat{X}(t_{k}) -\widetilde{X}(t_{k})$ is independent of $\Delta W_k^{(l)}$ and $I_k(j,l)$, we obtain 
\begin{align*}
    \mathbb{E}\Big(\sum_{l=1}^m(B^{(l)})^2\big(\widehat{X}(t_k)-\widetilde{X}(t_k)\big)\big((\Delta W_k^{(l)})^2-\Delta t\big)\Big)^2 & = \sum_{l=1}^m\mathbb{E}\Big((B^{(l)})^2\big(\widehat{X}(t_k)-\widetilde{X}(t_k)\big)\Big)^2\mathbb{E}\Big((\Delta W_k^{(l)})^2-\Delta t\Big)^2 \\
    & \leq \tilde{C}_1\mathbb{E}\big\|\widehat{X}(t_k)-\widetilde{X}(t_k)\big\|^2 \Delta t^2,
\end{align*}
and 
\begin{align*}
     \mathbb{E}\Big( \sum_{\substack{l,j=1 \\l\neq j}}^m B^{(l)}B^{(j)}\big(\widehat{X}(t_k)-\widetilde{X}(t_k)\big)I_k(j,l)\Big)^2 & = \sum_{\substack{l,j=1 \\l\neq j}}^m \mathbb{E}\big\|B^{(l)}B^{(j)}\big(\widehat{X}(t_k)-\widetilde{X}(t_k)\big)\big\|^2\mathbb{E}(I_k(j,l))^2 \\
     & \leq \tilde{C}_2 \mathbb{E}\big\|\widehat{X}(t_k)-\widetilde{X}(t_k)\big\|^2 \Delta t^2.
\end{align*}

\begin{align*}
    \left|\frac{1}{4}\mathbb{E}\Big(\sum_{\substack{l,j=1 \\l\neq j}}^m\big(R_k^{(l,j)}+(R_k^{(l,j)})^\dagger\big)(I_k(j,l))^2\Big)\right| & \leq \frac{1}{4} \sum_{\substack{l,j=1 \\l\neq j}}^m\mathbb{E}\big|R_k^{(l,j)}+(R_k^{(l,j)})^\dagger\big|\mathbb{E}(I_k(j,l))^2 \\
    & \leq \tilde{C}_3 \sqrt{\mathbb{E}\big\|\widehat{X}(t_k)-\widetilde{X}(t_k)\big\|^2}\Delta t^2.
\end{align*}
Let $\epsilon_k=\sqrt{\mathbb{E}\big\|\widehat{X}(t_k) -\widetilde{X}(t_k)\big\|^2} $, then we know 
\begin{align*}
    \epsilon_{k+1}^2 \leq \epsilon_k^2+C_1 \epsilon_k^2 \Delta t +C_2 \epsilon_k^2 \Delta t^2+C_3 \epsilon_k \Delta t^2 + C_4 \Delta t^2 + C_5 \Delta t^3, \quad \epsilon_{0}=0.
\end{align*}
 Here $\tilde{C}_1, \tilde{C}_2,\tilde{C}_3, C_1,C_2,C_3,C_4$ and $C_5$ are positive constants independent of  $\Delta t$ and $k$. 
 
 We will prove by induction that $\epsilon_k^2 = \mathcal{O}(\Delta t^2)$. Consider an auxiliary equation
\begin{equation}
    \eta_{k+1}^2 = \eta_k^2+C_1 \eta_k^2 \Delta t +C_2 \eta_k^2 \Delta t^2+C_3 \eta_k \Delta t^2 + C_4 \Delta t^2 + C_5 \Delta t^3, \quad \quad \eta_{0}=0.
\end{equation}
It is clear that $\epsilon_k \leq \eta_k$ and $\eta^2_k=\mathcal{O}(\Delta t^2)$ implies $\epsilon^2_k=\mathcal{O}(\Delta t^2)$. For every integer $1 \leq m \leq N_T$, assume for $k=0,1,\cdots, m-1$, $\eta^2_k=\mathcal{O}(\Delta t^2)$, we will prove $\eta^2_m=\mathcal{O}(\Delta t^2)$. As $\eta^2_k$ is an increasing sequence, there exits a constant $c>0$ and $k_0 < m-1$ that $\eta^2_{k_0}\leq c \Delta t^2$ and $\eta^2_{k_0+1}\geq c \Delta t^2$. For $k_0 < k \leq m-1$, 
\begin{equation*}
    \eta^2_{k+1}\leq \eta^2_k\left(\frac{C_4}{c}+1 + (C_1+\frac{C_3}{\sqrt{c}}+\frac{C_5}{c})\Delta t + C_2 \Delta t^2\right).
\end{equation*}
Then, by the fact that  $\eta^2_{k_0+1}=\mathcal{O}(\Delta t^2)$, we know
\begin{equation*}
    \eta^2_m \leq \eta^2_{k_0+1}\left(\frac{C_4}{c}+1 + (C_1+\frac{C_3}{\sqrt{c}}+\frac{C_5}{c})\Delta t + C_2 \Delta t^2\right)^m \leq C \Delta t^2.
\end{equation*}
Here, constant $C$ are irrelevant with $m$ and $\Delta t$. 
We have proven that
\begin{equation*}
    \mathbb{E}\big(\widehat{X}(T) -\widetilde{X}(T)\big)^2 = \mathcal{O}(\Delta t^2).
\end{equation*}

With the order $1$-strong convergence of the Milstein method under the mean square norm (see, for example, \cite[Theorem 10.3.5]{kloeden1992}), we conclude the approximate equation \eqref{MultiApproProof} strongly converges to the solution of the SDE \eqref{MultiProof} with at least order $1$.
\end{proof}

\section{Summary}\label{summary}

We proposed quantum algorithms for linear stochastic differential equations (SDEs) with both Gaussian noise and $\alpha$-stable L\'evy noise. The algorithms are based on solving approximate equations to the corresponding SDEs by the Schr\"odingerisation method. Compared to the Milstein scheme, the gates complexity of our algorithms exhibits a polylogarithmic dependence on both the dimensions and the number of samples, which demonstrates the quantum advantage over their classical counterparts in scenarios involving high dimensions and large sample sizes. We also obtained the convergence rate of the algorithms for SDEs with Gaussian noise in the mean square norm. 

Our algorithms were applied to several examples, including Ornstein-Uhlenbeck processes, geometric Brownian motions and L\'evy flights. It can be extended to general linear SDEs and reducible SDEs. Similarly as most current stage quantum algorithms, our approach can only solve linear SDEs. In practice, we also need to carefully choose the recovery region to avoid introducing additional error. Further development may be addressed to nonlinear problems. Since the gate complexity of our approach exhibits only a polylogarithmic dependence on dimension, it may also serve as an efficient approximate method for computing solutions to stochastic partial differential equations. It is also well-positioned for applications in areas such as Monte Carlo simulations and backward SDEs.

\bigskip

{\bf Acknowledgements}
SJ and NL  were  supported by NSFC grant No. 12341104, the Shanghai Jiao Tong University 2030 Initiative and the Fundamental Research
Funds for the Central Universities. SJ was also partially supported by the NSFC grants No. 12031013.  NL acknowledges funding from the Science and Technology Program of Shanghai, China (21JC1402900). NL is also supported by NSFC grants No.12471411, the Shanghai Jiao Tong University 2030 Initiative, and the Fundamental Research Funds for the Central Universities.
\bigskip

{\bf Data Availability } 
Data availability is not applicable to this article as no new data were generated or analyzed in this research.

\section*{Declarations}

{\bf Conflict of interest} No potential conflict of interest was declared by the authors.

\bibliographystyle{plain}
\bibliography{references}

\begin{thebibliography}{10}

\bibitem{BosonSampling}
Scott Aaronson and Alex Arkhipov.
\newblock The computational complexity of linear optics.
\newblock {\em Theory Comput.}, 9(4):143--252, 2013.

\bibitem{An2021timedependent}
Dong An, Di~Fang, and Lin Lin.
\newblock Time-dependent unbounded {H}amiltonian simulation with vector norm scaling.
\newblock {\em {Quantum}}, 5:459, 2021.

\bibitem{An2022timedependent}
Dong An, Di~Fang, and Lin Lin.
\newblock Time-dependent {H}amiltonian simulation of highly oscillatory dynamics and superconvergence for {S}chr\"{o}dinger equation.
\newblock {\em {Quantum}}, 6:690, 2022.

\bibitem{QAMLMC}
Dong An, Noah Linden, Jin-Peng Liu, Ashley Montanaro, Changpeng Shao, and Jiasu Wang.
\newblock Quantum-accelerated multilevel {M}onte {C}arlo methods for stochastic differential equations in mathematical finance.
\newblock {\em {Quantum}}, 5:481, 2021.

\bibitem{2ndRK}
Weizhu Bao and Shi Jin.
\newblock High-order i-stable centered difference schemes for viscous compressible flows.
\newblock {\em J. Comput. Math.}, 21(1):101--112, 2003.

\bibitem{Levyflightsoptics}
Pierre Barthelemy, Jacopo Bertolotti, and Diederik~S Wiersma.
\newblock A {L\'e}vy flight for light.
\newblock {\em Nature}, 453(7194):495--498, 2008.

\bibitem{GBMBlackScholes}
Fischer Black and Myron Scholes.
\newblock The pricing of options and corporate liabilities.
\newblock {\em J. Polit. Econ.}, 81(3):637--654, 1973.

\bibitem{QAMCinfinitevariance}
Jose Blanchet, Mario Szegedy, and Guanyang Wang.
\newblock Quadratic speed-up in infinite variance quantum {M}onte {C}arlo.
\newblock {\em arXiv:2401.07497}, 2024.

\bibitem{cao2023Timedependent}
Yu~Cao, Shi Jin, and Nana Liu.
\newblock Quantum simulation for time-dependent {Hamiltonians} -- with applications to non-autonomous ordinary and partial differential equations.
\newblock {\em J. Phys. A: Math. Theor.}, 58:155304, 2025.

\bibitem{UniversalQW}
Andrew~M. Childs.
\newblock Universal computation by quantum walk.
\newblock {\em Phys. Rev. Lett.}, 102:180501, 2009.

\bibitem{ExponentialQW}
Andrew~M. Childs, Richard Cleve, Enrico Deotto, Edward Farhi, Sam Gutmann, and Daniel~A. Spielman.
\newblock Exponential algorithmic speedup by a quantum walk.
\newblock In {\em Proceedings of the Thirty-Fifth Annual ACM Symposium on Theory of Computing}, page 59–68, 2003.

\bibitem{ChildsQASL2017}
Andrew~M. Childs, Robin Kothari, and Rolando~D. Somma.
\newblock Quantum algorithm for systems of linear equations with exponentially improved dependence on precision.
\newblock {\em SIAM J. Comput.}, 46(6):1920--1950, 2017.

\bibitem{Levyflightsearthquake}
Alvaro Corral.
\newblock Universal earthquake-occurrence jumps, correlations with time, and anomalous diffusion.
\newblock {\em Phys. Rev. Lett.}, 97(17):178501, 2006.

\bibitem{2021NeurIPSSchrodingerBridge}
Valentin De~Bortoli, James Thornton, Jeremy Heng, and Arnaud Doucet.
\newblock Diffusion {S}chr\"{o}dinger bridge with applications to score-based generative modeling.
\newblock In {\em Proceedings of the 35th International Conference on Neural Information Processing Systems}, pages 17695--17709, 2021.

\bibitem{Xiantao2024OpenQS}
Zhiyan Ding, Xiantao Li, and Lin Lin.
\newblock Simulating open quantum systems using {H}amiltonian simulations.
\newblock {\em PRX Quantum}, 5(2):020332, 2024.

\bibitem{2022ICLRLangevin}
Tim Dockhorn, Arash Vahdat, and Karsten Kreis.
\newblock Score-based generative modeling with critically-damped {L}angevin diffusion.
\newblock In {\em International Conference on Learning Representations}, 2022.

\bibitem{gardiner2009stochastic}
Crispin Gardiner.
\newblock {\em Stochastic Methods}, volume~4.
\newblock Springer Berlin, 2009.

\bibitem{OUphysics}
R\'emi Goerlich, Minghao Li, Samuel Albert, Giovanni Manfredi, Paul-Antoine Hervieux, and Cyriaque Genet.
\newblock Noise and ergodic properties of {Brownian} motion in an optical tweezer: {Looking} at regime crossovers in an {Ornstein-Uhlenbeck} process.
\newblock {\em Phys. Rev. E}, 103:032132, 2021.

\bibitem{QCforfinancePRR}
Javier Gonzalez-Conde, \'Angel Rodr\'{\i}guez-Rozas, Enrique Solano, and Mikel Sanz.
\newblock Efficient {Hamiltonian} simulation for solving option price dynamics.
\newblock {\em Phys. Rev. Research}, 5:043220, 2023.

\bibitem{Grover}
Lov~K. Grover.
\newblock A fast quantum mechanical algorithm for database search.
\newblock In {\em Proceedings of the Twenty-Eighth Annual ACM Symposium on Theory of Computing}, page 212–219, 1996.

\bibitem{HHL}
Aram~W. Harrow, Avinatan Hassidim, and Seth Lloyd.
\newblock Quantum algorithm for linear systems of equations.
\newblock {\em Phys. Rev. Lett.}, 103:150502, 2009.

\bibitem{hu2024qc}
Junpeng Hu, Shi Jin, Nana Liu, and Lei Zhang.
\newblock Quantum circuits for partial differential equations via {Schr\"odingerisation}.
\newblock {\em Quantum}, 8:1563, 2024.

\bibitem{Hu2024Multiscale}
Junpeng Hu, Shi Jin, and Lei Zhang.
\newblock Quantum algorithms for multiscale partial differential equations.
\newblock {\em Multiscale Model. Simul.}, 22(3):1030--1067, 2024.

\bibitem{2021NeurIPSvariational}
Chin-Wei Huang, Jae~Hyun Lim, and Aaron~C Courville.
\newblock A variational perspective on diffusion-based generative models and score matching.
\newblock In {\em Proceedings of the 35th International Conference on Neural Information Processing Systems}, pages 22863--22876, 2021.

\bibitem{OUbiology}
Gene Hunt.
\newblock The relative importance of directional change, random walks, and stasis in the evolution of fossil lineages.
\newblock {\em Proc. Natl. Acad. Sci. USA}, 104(47):18404--18408, 2007.

\bibitem{analog2023}
Shi Jin and Nana Liu.
\newblock Analog quantum simulation of partial differential equations.
\newblock {\em Quantum Sci. Technol.}, 9(3):035047, 2024.

\bibitem{Ma2023Maxwell}
Shi Jin, Nana Liu, and Chuwen Ma.
\newblock Quantum simulation of {Maxwell's} equations via schr{\"o}dingerisation.
\newblock {\em ESAIM: Math. Modell. Numer. Anal.}, 58(5):1853--1879, 2024.

\bibitem{Ma2024illpoesd}
Shi Jin, Nana Liu, and Chuwen Ma.
\newblock Schr\"odingerisation based computationally stable algorithms for ill-posed problems in partial differential equations.
\newblock {\em arXiv:2403.19123}, 2024.

\bibitem{Ma2024inhomov2}
Shi Jin, Nana Liu, and Chuwen Ma.
\newblock On {Schr\"odingerization} based quantum algorithms for linear dynamical systems with inhomogeneous terms.
\newblock {\em arXiv:2402.14696v2}, 2025.

\bibitem{jin2023quantum}
Shi Jin, Nana Liu, and Yue Yu.
\newblock Quantum simulation of partial differential equations: {Applications} and detailed analysis.
\newblock {\em Phys. Rev. A}, 108:032603, 2023.

\bibitem{jin2024heatqc}
Shi Jin, Nana Liu, and Yue Yu.
\newblock Quantum circuits for the heat equation with physical boundary conditions via {Schr\"odingerisation}.
\newblock {\em arXiv:2407.15895}, 2024.

\bibitem{Schrodingerisation1}
Shi Jin, Nana Liu, and Yue Yu.
\newblock Quantum simulation of partial differential equations via {Schr\"odingerization}.
\newblock {\em Phys. Rev. Lett.}, 133:230602, 2024.

\bibitem{jin2024qsFKP}
Shi Jin, Nana Liu, and Yue Yu.
\newblock Quantum simulation of the {Fokker-Planck} equation via {Schr\"odingerization}.
\newblock {\em arXiv:2404.13585}, 2024.

\bibitem{kloeden1992}
Peter~E Kloeden, Eckhard Platen, and Kloeden.
\newblock {\em Numerical Solution of Stochastic Differential Equations}.
\newblock Springer Berlin, Heidelberg, 1992.

\bibitem{SDENet}
Lingkai Kong, Jimeng Sun, and Chao Zhang.
\newblock {SDE-Net}: Equipping deep neural networks with uncertainty estimates.
\newblock In {\em Proceedings of the 37th International Conference on Machine Learning}, pages 5405--5415, 2020.

\bibitem{VariationalQCforSDE}
Kenji Kubo, Yuya~O. Nakagawa, Suguru Endo, and Shota Nagayama.
\newblock Variational quantum simulations of stochastic differential equations.
\newblock {\em Phys. Rev. A}, 103:052425, 2021.

\bibitem{kuhn2019strong}
Franziska K{\"u}hn and Ren{\'e}~L Schilling.
\newblock Strong convergence of the {E}uler--{M}aruyama approximation for a class of {L\'e}vy-driven {SDEs}.
\newblock {\em Stochastic Process. Appl.}, 129(8):2654--2680, 2019.

\bibitem{QuanEnhanceMCMC}
David Layden, Guglielmo Mazzola, Ryan~V. Mishmash, Mario Motta, Pawel Wocjan, Jin-Sung Kim, and Sarah Sheldon.
\newblock Quantum-enhanced {Markov chain Monte Carlo}.
\newblock {\em Nature}, 619(7969):282--287, 2023.

\bibitem{Levymodulus}
Paul~S. L\'evy.
\newblock Th{\'e}orie de l'addition des variables al{\'e}atoires.
\newblock {\em Math. Gaz.}, 39:344, 1955.

\bibitem{low2017optimal}
Guang~Hao Low and Isaac~L Chuang.
\newblock Optimal {Hamiltonian} simulation by quantum signal processing.
\newblock {\em Phys. Rev. Lett.}, 118(1):010501, 2017.

\bibitem{QAMC}
Ashley Montanaro.
\newblock Quantum speedup of {Monte Carlo} methods.
\newblock {\em Proc. R. Soc. A}, 471(2181):20150301, 2015.

\bibitem{QFTgatecomplexity}
Yunseong Nam, Yuan Su, and Dmitri Maslov.
\newblock Approximate quantum {Fourier} transform with {$O (n \log (n)) $ T} gates.
\newblock {\em NPJ Quantum Inf.}, 6(1):26, 2020.

\bibitem{QCforfinanceReview}
Román Orús, Samuel Mugel, and Enrique Lizaso.
\newblock Quantum computing for finance: {Overview} and prospects.
\newblock {\em Reviews in Physics}, 4:100028, 2019.

\bibitem{QAMCMCopt2025}
Guneykan Ozgul, Xiantao Li, Mehrdad Mahdavi, and Chunhao Wang.
\newblock Quantum speedups for {M}arkov {C}hain {M}onte {C}arlo methods with application to optimization.
\newblock {\em arXiv:2504.03626}, 2025.

\bibitem{Shor}
Peter~W. Shor.
\newblock Polynomial-time algorithms for prime factorization and discrete logarithms on a quantum computer.
\newblock {\em SIAM J. Comput.}, 26(5):1484--1509, 1997.

\bibitem{song2021scorebased}
Yang Song, Jascha Sohl-Dickstein, Diederik~P Kingma, Abhishek Kumar, Stefano Ermon, and Ben Poole.
\newblock Score-based generative modeling through stochastic differential equations.
\newblock In {\em International Conference on Learning Representations}, 2021.

\bibitem{QWreview}
Salvador~El{\'\i}as Venegas-Andraca.
\newblock Quantum walks: a comprehensive review.
\newblock {\em Quantum Inf. Process.}, 11(5):1015--1106, 2012.

\bibitem{Levyflightssearch}
G.M Viswanathan, V~Afanasyev, Sergey~V Buldyrev, Shlomo Havlin, M.G.E {da Luz}, E.P Raposo, and H.Eugene Stanley.
\newblock L\'evy flights in random searches.
\newblock {\em Phys. A}, 282(1):1--12, 2000.

\bibitem{Watkins2022Timedependent}
Jacob Watkins, Nathan Wiebe, Alessandro Roggero, and Dean Lee.
\newblock Time dependent {Hamiltonian} simulation using discrete clock constructions.
\newblock {\em PRX Quantum}, 5(4):040316, 2024.

\bibitem{NeurIPS2020SGD}
Pan Zhou, Jiashi Feng, Chao Ma, Caiming Xiong, Steven Chu~Hong Hoi, and Weinan E.
\newblock Towards theoretically understanding why sgd generalizes better than {Adam} in deep learning.
\newblock In {\em Proceedigs of the 34th International Conference on Nerual Information Processing Systems}, pages 21285--21296, 2020.

\end{thebibliography}

\end{document}